\documentclass[a4paper,onecolumn,11pt]{quantumarticle}
\pdfoutput=1
\usepackage[utf8]{inputenc}
\usepackage[english]{babel}
\usepackage[T1]{fontenc}
\usepackage{xcolor}
\usepackage{amsmath}
\usepackage{mathtools}
\usepackage{amssymb}
\usepackage[most]{tcolorbox}
\usepackage{amsthm}
\usepackage{tikz-cd}
\usepackage[colorlinks=true
,urlcolor=teal
,anchorcolor=blue
,citecolor=teal
,filecolor=blue
,linkcolor=teal
,menucolor=blue
,linktocpage=true
,pdfproducer=medialab
,pdfa=true
]{hyperref}

\newtheorem{thm}{Theorem}
\newtheorem{lem}{Lemma}
\newtheorem{defn}{Definition}
\newtheorem{prop}{Proposition}
\newtheorem{rem}{Remark}

\newcommand{\dd}{\mathrm{d}}
\newcommand{\ket}[1]{| #1 \rangle}
\newcommand{\bra}[1]{\langle #1 |}
\newcommand{\ketbra}[2]{\left|#1\right\rangle\!\!\left\langle #2\right|}

\def\h{\mathcal{H}}
\def\id{\mathbb{I}}
\newcommand*{\Tr}{\mathrm{Tr}}
\def\D{\mathcal{D}}
\def\E{\mathcal{E}}
\def\C{\mathcal{C}}
\def\map{\mathcal}

\begin{document}

\title{Quantum minimum description of density matrices}

\author{Patrick Hayden}
\affiliation{Leinweber Institute for Theoretical Physics, Stanford University, Stanford, CA 94305, USA}

\author{Alexander Maloney}
\affiliation{Department of Physics, Syracuse University, Syracuse, NY 13244, USA}
\affiliation{Institute for Quantum and Information Sciences, Syracuse University, Syracuse, NY 13244, USA}

\author{Jinzhao Wang}
\email{wang.jinzhao226@gmail.com}
\affiliation{Leinweber Institute for Theoretical Physics, Stanford University, Stanford, CA 94305, USA}
\affiliation{Zyphra, San Francisco, CA 94105, USA}

\author{Yuxiang Yang}
\affiliation{QICI Quantum Information and Computation Initiative, School of Computing and Data Science, The University of Hong Kong, Pokfulam Road, Hong Kong, China}
 
\date{}
\maketitle

\begin{abstract}
We study a variant of Schumacher compression. The task asks for the minimal memory cost for compressing many copies of a density matrix with known spectrum and unknown eigenbasis, without preserving its purification. For fixed dimension and distinct positive eigenvalues, we obtain the cost through its additive constant. Achievability uses a generalization of Werner's cloning map to $\mathrm{GL}(d,\mathbb C)$ irreducible representations, with finite trace-distance bounds controlled by highest-weight differences and row gaps. The matching converse follows from a quantitative form of Koashi--Imoto incompressibility for irreducible group orbits under a spectral-gap assumption. We also identify the relation between this quantum minimum description length and universal lossless coding overhead, and its connection to free entropy is explained in a companion letter~\cite{paper1}. We provide a \href{https://github.com/JWang226/Quantum-Minimum-Description-Length}{Lean certificate} for our proofs. 
\end{abstract}

\section{Introduction}\label{sec:intro}

Schumacher compression characterizes the memory needed to preserve quantum information together with its correlations with a reference~\cite{schumacher1995quantum}. A different problem arises when one only needs to reconstruct many identical copies of a state whose eigenbasis is unknown. The memory must retain this unknown basis information, but need not preserve a purification. We determine the optimal memory cost through its additive constant and construct a protocol attaining it.

The task is to compress identical and independent copies of a density matrix whose spectrum $x=(x_1,\ldots,x_d)$ is known but whose eigenbasis is not. Fix a reference state $\rho$ with this spectrum and write $\rho_g:=U(g)\rho U(g)^\dagger$ for $g\in\mathrm U(d)$. Let $M_n$ be the sole retained memory, including any classical register, and write $|M_n|:=\log\dim M_n$; all logarithms are base two. The encoder $\E_n:\mathcal B((\mathbb C^d)^{\otimes n})\to\mathcal B(M_n)$ and decoder $\D_n:\mathcal B(M_n)\to\mathcal B((\mathbb C^d)^{\otimes n})$ are completely positive trace-preserving maps, which may depend on $n$ and $x$ but not on~$g$. They form a $(|M_n|,\delta)$-code if
\begin{equation}\label{eq:error}
   \sup_{g\in\mathrm U(d)}\frac12
   \|\rho_g^{\otimes n}-\D_n\circ\E_n(\rho_g^{\otimes n})\|_1\le\delta\ .
\end{equation} 

We seek the smallest asymptotic memory size among sequences with $\delta_n\to0$. We call this cost the \emph{quantum minimum description length} (QMDL) of the known-spectrum orbit family. A completely known state can be regenerated without storing its input; the unknown eigenbasis makes the present task nontrivial. For a nontrivial orbit, a purely classical strategy such as tomography followed by state preparation cannot achieve vanishing global error~\cite{yang2018compression}. The cost $|M_n|$ is the log dimension, rather than the rounded number $\lceil\log\dim M_n\rceil$ of physical qubits; this distinction matters for the additive constant.

\begin{thm}[Optimal memory cost]\label{thm:qmdl}
Fix $d$ and a rank-$r$ density matrix with distinct positive eigenvalues $x_1>\cdots>x_r>0$. Define
\begin{align}
L_{d,r}(n,x):=\frac{r(2d-r-1)}2\log n
+\sum_{1\le i<j\le r}\log(x_i-x_j)+(d-r)\sum_{i=1}^r\log x_i
-\sum_{k=d-r}^{d-1}\log k! .\label{eq:result}
\end{align}
There exists a sequence of codes with $|M_n|=L_{d,r}(n,x)+o(1)$ and $\delta_n=O(\log n/\sqrt n)$. Every sequence of codes with $\delta_n\to0$ satisfies
\[
\liminf_{n\to\infty}\bigl[|M_n|-L_{d,r}(n,x)\bigr]\ge0.
\]
The lower bound also holds when only the Haar-average reconstruction
error tends to zero.
\end{thm}
Theorem~\ref{thm:qmdl} combines the achievability and converse results proved in Sections~\ref{sec:direct} and~\ref{sec:converse}. The dimension and spectrum are fixed throughout these limits; no uniform claim is made as positive eigenvalues collide or approach zero. The theorem asserts a sequence of codes that achieves this expansion with vanishing error, rather than the same expansion for every prescribed error schedule.

Earlier work established the leading memory cost~\cite{yang2016optimal,yang2018compression}. For qubits, the protocol of Yang, Chiribella, and Hayashi (YCH) already has an achievable spectrum-dependent constant. We prove the matching converse through that constant and extend the construction to qudits. The coefficient of $\log n$ counts half the real dimension of the orbit, while the additive term resolves the eigenvalue dependence.

The construction uses covariant channels between
$\mathrm{GL}(d,\mathbb C)$ irreducible representation spaces.
Their description by normalized Choi projectors is established
representation-theoretic machinery
\cite{gschwendtner2021programmability,lee2022quantum,aschieri2024equivariant}.
We select the Parthasarathy--Ranga Rao--Varadarajan (PRV) component,
obtaining channels that extend Werner's symmetric-power
cloner~\cite{werner1998optimal}. For $d\ge2$ and partitions
$\mu,\nu$ with at most $r$ rows and dominant difference
$\omega=\nu-\mu$, Theorem~\ref{thm:main} gives
the finite trace-distance bound $C_{d,x}D/(b_\mu+1)$ in both
directions, where $D=\|\omega\|_1$ and $b_\mu$ is the minimum of
the first $\min(r,d-1)$ adjacent-row gaps of $\mu$.
A padded target representation stores all typical Schur blocks in
the same memory space. For each such block and this target,
$b_\mu=\Omega(n)$ and $D=O(\sqrt n\log n)$, giving error
$O(\log n/\sqrt n)$.

The converse uses a quantitative incompressibility bound for
irreducible compact-group orbits
(Proposition~\ref{prop:compact_orbit_memory}). If the orbit state
on $\mathcal H$ has a simple largest eigenvalue separated from the
remaining spectrum by $\gamma$, then any code with Haar-average
error $\overline\delta$ requires
$\dim M\ge\dim\mathcal H(1-\overline\delta/\gamma)$.
This bound applies to arbitrary encoding and decoding channels.
A uniform lower bound on $\gamma$ for the representation states
lets us retain the additive memory constant when transferring
the bound to the original $n$-copy family.

Universal lossless quantum data compression uses one code for every
state in the source family. Even with a known spectrum, the
unknown eigenbasis increases the required average code length.
We show that the smallest worst-case excess in ideal average length
above $nS(\rho)$ is $L_{d,r}(n,x)-s_{\mathrm{blk}}(x)+o(1)$.
Thus QMDL also quantifies this coding overhead, up to an explicit
entropy correction: $s_{\mathrm{blk}}(x)$ is the limiting entropy
of the irreducible representation state in a Schur block.
Section~\ref{sec:redundancy} identifies the optimal code state;
Appendix~\ref{app:unknown_spectrum} evaluates the additional
classical penalty when the spectrum is unknown.

The spectral terms in the memory cost also have a geometric
interpretation. For the rank-$r$ states in Theorem~\ref{thm:qmdl},
the unitary orbit $\Omega=\{U\rho U^\dagger\}$ has real dimension
$r(2d-r-1)$. Its volume in the induced Hilbert--Schmidt metric is
\[
\operatorname{Vol}_{\mathrm{HS}}(\Omega)
=c_{d,r}\prod_{1\le i<j\le r}(x_i-x_j)^2
\prod_{i=1}^r x_i^{2(d-r)},
\]
where $c_{d,r}$ depends only on $d,r$. One half of its logarithm
therefore reproduces the spectral terms in $L_{d,r}(n,x)$.
Our companion letter~\cite{paper1} makes the comparison precise
using a Hilbert--Schmidt tube-volume definition of physical free
entropy and relates it to Voiculescu's microstate free
entropy~\cite{voiculescu1993analogues,voiculescu1994analogues,voiculescu2002free}.
Asymptotically, the optimal memory cost equals one half of that
finite-dimensional quantity at resolution $n^{-1}$, up to an
explicit offset depending only on $d,r$. This is a comparison
between geometric volume and representation dimension; the
resolution belongs to the geometric counting definition.

Section~\ref{sec:related_works} reviews prior compression results,
and Section~\ref{sec:review} recalls the qubit protocol.
Section~\ref{sec:cloner} defines the PRV channels and their structural
properties. Sections~\ref{sec:direct} and~\ref{sec:converse} prove the
optimal memory theorem. Section~\ref{sec:redundancy} relates QMDL to
the extra code length required by universal lossless compression.
The technical trace-distance estimate is proved in
Section~\ref{sec:fidelity}.

The proofs of our main theorems are formalized in Lean and Comparator checked. See the \href{https://github.com/JWang226/Quantum-Minimum-Description-Length}{Lean Certificate} and an accompanying \href{https://jwang226.github.io/Quantum-Minimum-Description-Length/}{Proof Wiki guide}.

\section{Background}\label{sec:background} 
\subsection{Related work}\label{sec:related_works}
Quantum minimum description length connects the classical idea of
description length~\cite{rissanen1978modeling} with compression of
repeated quantum states. Plesch and Bu\v{z}ek~\cite{plesch2010efficient}
gave an $O(\log n)$-qubit compression scheme for $n$ identical pure
qubits. The reduction uses their support on the symmetric subspace,
whose dimension grows polynomially with $n$. Rozema et al.~implemented
the compression of three pure qubits into two~\cite{rozema2014quantum}.
For mixed states, the algebraic basis of exact compression is given
by Koashi and Imoto~\cite{koashi2001compressibility,koashi2001possible}.

The YCH scheme~\cite{yang2016optimal} established the leading memory
cost for mixed qubits and an achievable spectrum-dependent constant.
It also underlies the ``quantum stopwatch'' protocols
of~\cite{yang2018quantum,yang2018compressionfor}.
Hayashi and Tan~\cite{hayashi2017minimum} determined the optimal
compression of diagonal states. For qudits, the $O(\log n)$ scheme
of~\cite{yang2016efficient} had a suboptimal prefactor;
Yang et al.~\cite{yang2018compression} subsequently obtained the
optimal leading cost, with roughly $\frac12\log n$ per unknown
parameter. Their memory combines quantum and classical registers;
a purely classical memory does not suffice for unknown-basis families.

Further work treats states generated by shallow circuits with both
$n$ and $d$ growing~\cite{yang2025compression}, and experimental
implementations appear in~\cite{pivoluska2022implementation,xu2024experimental}.
Here we determine the optimal additive constant for fixed spectra
with distinct positive eigenvalues and extend the YCH construction
to qudits.

\subsection{The YCH qubit protocol}\label{sec:review}
The YCH protocol~\cite{yang2016optimal} compresses $n$ copies of
$\rho=U(g)\rho_0U(g)^\dagger$, where $g\in\mathrm U(2)$ is unknown
and $\rho_0=\operatorname{diag}(p,1-p)$ with $1/2<p\le1$.
At $p=1/2$ the orbit is a single known state and needs no memory.

The Schur transform resolves the permutation symmetry:
\begin{equation}
    V_\mathrm{Schur}: (\mathbb{C}^2)^{\otimes n} \to \bigoplus_{J=j_{\min}}^{n/2} \h_J \otimes \map{M}_J
\end{equation}
Here $\h_J$ is the $\mathrm{GL}(2,\mathbb C)$ irrep of dimension
$2J+1$ labeled by $(n/2+J,n/2-J)$, and $\map M_J$ is the
corresponding symmetric-group irrep. The allowed spins start at
$j_{\min}=0$ for even $n$ and $j_{\min}=1/2$ for odd $n$.

Applying the Schur transform to the input state block-diagonalizes it:
\begin{equation}\label{eq:schur}
    V_\mathrm{Schur} \rho^{\otimes n} V_\mathrm{Schur}^\dagger = \bigoplus_{J=j_{\min}}^{n/2} q_J \rho_{g,J} \otimes \tau_J
\end{equation}
where $\rho_{g,J}:=U_J(g)\rho_JU_J(g)^\dagger$ and $\tau_J$ is
maximally mixed on $\map M_J$. Writing $q:=1-p$, the probability
of the $J$ sector is
\begin{equation}\label{eq:qJ}
    q_J = \dim\map{M}_J\cdot s_{(\frac n2+J,\,\frac n2-J)}(p,q)
    = \frac{2J+1}{\frac n2+J+1}\binom{n}{\frac n2-J}\sum_{j=\frac n2-J}^{\frac n2+J} p^{\,j}q^{\,n-j}\ ,
\end{equation}
where $s_\lambda$ is the Schur polynomial and
$\dim\map M_J=\frac{2J+1}{n/2+J+1}\binom{n}{n/2-J}$.

The protocol stores every measured sector in one target irrep:
\begin{tcolorbox}[colback=gray!5!white, colframe=gray!100!black, title=Qubit Compression Scheme (YCH~\cite{yang2016optimal})]
\textbf{Encoding:}
\begin{enumerate}
    \item \textbf{Irrep Measurement:} Perform a non-demolition measurement of the quantum number $J$, preserving the quantum information within the associated subspace $\h_J \otimes \map{M}_J$. 
    \item \textbf{Cloning:} Choose $J^\star$ to be the allowed spin in $\{j_{\min},j_{\min}+1,\ldots,n/2\}$ nearest the concentration center $(p-1/2)(n+1)$. This need not be the exact mode of $q_J$; at $p=1$ it gives $J^\star=n/2$. Discard the multiplicity register $\map{M}_J$ and apply the cloning channel $\C_{J\to J^\star}$ to map the remaining state $\rho_{g,J}$ to a fixed irrep $J^\star$. Store this output state into a quantum memory of dimension $2J^\star + 1$.
\end{enumerate}

\textbf{Decoding:}
\begin{enumerate}
    \item \textbf{Irrep Sampling:} Pick a target value $K$ at random from the distribution $q_K$.
    \item \textbf{Cloning:} Apply the optimal cloning channel $\mathcal{C}_{J^\star\to K}$ to the state in the quantum memory.
    \item \textbf{State Reconstruction:} Append a new multiplicity register initialized in the maximally mixed state $\tau_K$, and perform the inverse Schur transform.
\end{enumerate}
\end{tcolorbox}

Only the irrep state carries the unknown eigenbasis. The multiplicity
register can be discarded and regenerated, while the known distribution
$q_J$ can be resampled. Cloning into a single target irrep removes the
need to retain $J$; concentration near $J^\star$ controls the error.

The channel uses Werner's optimal approximate cloner from $n$ to
$m\ge n$ identical pure states~\cite{werner1998optimal}:
\begin{equation}\label{eq:werner}
    \mathcal{C}_{n \to m}:\mathrm{Sym}^n(\mathbb{C}^d)\to\mathrm{Sym}^m(\mathbb{C}^d),\quad \rho\mapsto \frac{d_{\mathrm{Sym}^n(\mathbb{C}^d)}}{d_{\mathrm{Sym}^m(\mathbb{C}^d)}}V^\dagger(\rho\otimes \id_{m-n})V
\end{equation}
where $\id_{m-n}$ is the identity on $\mathrm{Sym}^{m-n}(\mathbb{C}^d)$, and $V$ embeds $\mathrm{Sym}^m(\mathbb{C}^d)$ isometrically into $\mathrm{Sym}^n(\mathbb{C}^d)\otimes \mathrm{Sym}^{m-n}(\mathbb{C}^d)$.\footnote{The cloning map is usually written as $\mathcal{C}_{n \to m}(\rho) = \frac{d_{\mathrm{Sym}^n(\mathbb{C}^d)}}{d_{\mathrm{Sym}^m(\mathbb{C}^d)}}P_{\rm{Sym}^m}(\rho\otimes \id_{m-n})P_{\rm{Sym}^m}$, where $P_{\rm{Sym}^m}$ projects onto $\mathrm{Sym}^m(\mathbb{C}^d)$ inside $\mathrm{Sym}^n(\mathbb{C}^d)\otimes \mathrm{Sym}^{m-n}(\mathbb{C}^d)$. We use the intrinsic isometry $V$ so that the codomain is $\mathrm{Sym}^m(\mathbb{C}^d)$.} It is fixed up to a phase by
\begin{equation}
V U(g)^{\otimes m}
=(U(g)^{\otimes n}\otimes U(g)^{\otimes(m-n)})V,
\qquad g\in\mathrm U(d).
\end{equation}

As an $\mathrm{SU}(2)$ representation,
$\h_J\simeq\mathrm{Sym}^{2J}(\mathbb C^2)$; determinant factors
cancel under conjugation. Werner's cloner therefore maps spin $J$
to spin $K\ge J$, while partial trace handles $J>K$:
\begin{equation}\label{eq:qubitcloning}
    \mathcal{C}_{J\to K}(\rho_J) :=  \begin{cases} 
        \frac{2J+1}{2K+1} V^\dagger(\rho_J \otimes \id_{2K-2J})V, & J \le K \\ 
        \Tr_{2J-2K}(\rho_J), & J > K
    \end{cases}\ .
\end{equation}

The intertwining property of $V$ and unitary invariance of the
auxiliary identity make this channel covariant:
\begin{equation}
    \mathcal{C}_{J\to K}\left(U_J(g) \rho_J U_J^\dagger(g)\right) = U_K(g) \mathcal{C}_{J\to K}(\rho_J) U_K^\dagger(g)\ .
\end{equation}

YCH demonstrated that this compression scheme achieves asymptotically vanishing error. Since $J^\star=(p-1/2)(n+1)+O(1)$, its achievable memory cost is
\begin{equation}\label{eq:achievable_qubit}
    |M_n|=\log (2J^\star+1) = \log n + \log (2p-1) + o(1)\ .
\end{equation}

\section{Generalized cloning channels}\label{sec:cloner}

For qudits, Schur--Weyl duality gives
\begin{equation}
    V_\mathrm{Schur}: (\mathbb C^d)^{\otimes n}\to\bigoplus_\lambda\h_\lambda\otimes\map{M}_\lambda
\end{equation}
where $\lambda\vdash n$ is a partition with
$\lambda_1\ge\cdots\ge\lambda_d\ge0$,
$\h_\lambda$ is the corresponding $\mathrm{GL}(d,\mathbb C)$ irrep
of dimension $d_\lambda$, and $\map M_\lambda$ is the symmetric-group irrep.

For a density matrix $\rho$ on $\mathbb{C}^d$, the Schur decomposition is
\begin{equation}\label{eq:schur_d}
    V_\mathrm{Schur}\rho^{\otimes n}V_\mathrm{Schur}^\dagger
    = \bigoplus_{\lambda\vdash n} q_{\lambda,n}\,
    \rho_\lambda\otimes\tau_\lambda\ ,
\end{equation}
where $\rho_\lambda$ is the normalized $\mathrm{GL}(d,\mathbb{C})$ irrep
state of $\rho$, $\tau_\lambda:=\id_{\map M_\lambda}/\dim\map M_\lambda$
is the maximally mixed state on $\map M_\lambda$, and
\begin{equation}\label{eq:qlambda_def}
    q_{\lambda,n} := s_\lambda(x)\dim\map{M}_\lambda\ ,
\end{equation}
with $x$ the spectrum of $\rho$ and $s_\lambda$ the Schur polynomial.

Following YCH, we measure the Schur block label and obtain some outcome $\lambda$. The key is to find a $\mathrm{U}(d)$-covariant channel that maps every resulting irrep state to one fixed target irrep; Section~\ref{sec:direct} chooses a padded target that dominates the typical set.

\subsection{Generalized cloner from the PRV component}\label{sec:definition}
We use finite-dimensional rational irreducible representations of
$\mathrm{GL}(d,\mathbb C)$, restricted to $\mathrm U(d)$ when discussing
covariance. Their highest weights are dominant integral vectors and may
have negative entries. Partitions label the polynomial representations,
which extend to all matrices in $M_d(\mathbb C)$. The input and output
weights in the state-compression theorem are partitions; auxiliary and
dual representations may have general dominant integral highest weights.

More explicitly, a highest weight $\lambda$ gives
$\pi_\lambda=\det\nolimits^{\lambda_d}\otimes
\pi_{\lambda-\lambda_d(1,\ldots,1)}$. Thus
$\pi_{(1,-1)}=\det\nolimits^{-1}\otimes\pi_{(2,0)}$ for
$\mathrm{GL}(2)$. A determinant twist cancels from the unitary
conjugation action and from the normalized representation state of an
invertible density matrix. It need not extend to singular matrices:
negative determinant powers require invertibility. We therefore retain
general dominant auxiliary weights and state the domains explicitly.

A channel $\map N:\mathcal B(\h_\mu)\to\mathcal B(\h_\nu)$ between
irreducible representation spaces is $\mathrm U(d)$-covariant if
\begin{equation}\label{eq:covariance}
    \map{N}\circ\map{U}_\mu(g) = \map{U}_\nu(g)\circ\map{N},\quad\forall g\in\mathrm{U}(d)\ ,
\end{equation}
where $\map U_\mu(g)(X):=U_\mu(g)XU_\mu(g)^\dagger$, and similarly
for $\nu$; see Figure~\ref{fig:covariance_diagram}.

\begin{figure}
    \centering
    \begin{tikzcd}[row sep=large, column sep=large]
        \mathcal{B}(\mathcal{H}_\mu) \arrow[r, "\mathcal{N}"] \arrow[d, "\mathcal{U}_\mu(g)"'] & \mathcal{B}(\mathcal{H}_\nu) \arrow[d, "\mathcal{U}_\nu(g)"] \\
        \mathcal{B}(\mathcal{H}_\mu) \arrow[r, "\mathcal{N}"'] & \mathcal{B}(\mathcal{H}_\nu)
    \end{tikzcd}
    \caption{Commutative diagram for the covariance condition~\eqref{eq:covariance}: applying $\mathcal U_\mu(g)$ before the channel equals applying $\mathcal U_\nu(g)$ after it.}
    \label{fig:covariance_diagram}
\end{figure}

Equivalently, the Choi operator satisfies
\cite[Lemma~11]{gschwendtner2021programmability}
\begin{equation}
    [J_{\mathcal{N}},U_\mu^*\otimes U_\nu]=0\ .
\end{equation}
The tensor product of the dual representation $\mu^*$ with $\nu$ decomposes as
\begin{equation}\label{eq:decomp}
    \mathcal H_{\mu^*}\otimes\mathcal H_\nu
    \simeq \bigoplus_\lambda \mathcal H_\lambda\otimes\mathbb C^{m_\lambda}\ .
\end{equation}
Schur's lemma therefore gives the general commutant form
\begin{equation}\label{eq:covChoi}
    J_\mathcal N\simeq
    \bigoplus_\lambda \id_{\mathcal H_\lambda}\otimes A_\lambda,
    \qquad A_\lambda\ge0\ ,
\end{equation}
with the trace-preserving condition imposing additional linear constraints on the matrices $A_\lambda$. In a multiplicity-one sector, the normalized projector $(d_\mu/d_\lambda)\Pi_\lambda$ is itself the Choi operator of a covariant channel: its partial trace over $\mathcal H_\nu$ is $\id_{\mu^*}$ by Schur's lemma and trace normalization. The construction below uses precisely such a sector.
The structure in~\eqref{eq:covChoi} may be viewed as a simpler, irrep-level version of the more general covariant-channel classification in~\cite{mancinska2025classification}.

We choose the projector onto the PRV component, denoted by $\omega$.
Its Cartan intertwiner formula and the approximation bounds in
Theorem~\ref{thm:main} provide the properties needed for compression.
The representation is minimal in the dominance order, which for equal
total weights $|\lambda|:=\sum_i\lambda_i$ is defined by
\begin{equation}\label{eq:dominance}
    \lambda\preceq\eta
    \quad\Longleftrightarrow\quad
    |\lambda|=|\eta|\ \text{ and }\
    \sum_{i=1}^k\lambda_i\le\sum_{i=1}^k\eta_i
    \quad(1\le k<d)\ .
\end{equation}

The tensor product has a unique minimum in this order, occurring with
multiplicity one: the \emph{Parthasarathy--Ranga Rao--Varadarajan
(PRV) component}~\cite{parthasarathy1967,kumar1988proof}.
Since the dual highest weight is $\mu^*=-w_0(\mu)$, it is
\begin{equation}
    \omega=(\mu^*+w_0(\nu))^+=(\nu-\mu)^+.
\end{equation}
Here $w_0$ reverses coordinates, and $(\lambda)^+$ sorts them into
dominant order $\lambda_1\ge\cdots\ge\lambda_d$.
In the Cartan case, $\nu-\mu$ is already dominant and
$\omega=\nu-\mu$.

For the state comparisons in this Cartan case, $\mu$ and $\nu$ are
partitions with at most $r=\operatorname{rank}\rho$ rows. When $r=d$,
the auxiliary highest weight $\omega$ may have negative entries; its
finite-dimensional rational representation can be evaluated at $\rho$
because $\rho$ is invertible. When $r<d$, we have
$\omega_{r+1}=\cdots=\omega_d=0$, so dominance forces $\omega$ to be
a partition.

We therefore define our \emph{generalized cloning map} as follows.

\begin{defn}[Generalized Cloning Map]
Let $\mu$ and $\nu$ be $\mathrm{GL}(d,\mathbb{C})$ irreps, and let $\omega = (\nu-\mu)^+$ be the PRV component in the decomposition of $\mu^*\otimes\nu$. The generalized covariant cloning map $\mathcal{C}_{\mu\to\nu}:\mathcal B(\mathcal{H}_\mu)\to \mathcal B(\mathcal{H}_\nu)$ is defined as
\begin{equation}\label{eq:gen_cloner}
    \mathcal{C}_{\mu\to\nu}(\sigma_\mu)
    :=\Tr_\mu\!\left[J_\omega
    (\sigma_\mu^\top\otimes\id_\nu)\right],
    \qquad J_\omega:=\frac{d_\mu}{d_\omega}\Pi_\omega\ .
\end{equation}
\end{defn}

Multiplicity one makes $\Pi_\omega$ unambiguous. The compression
analysis uses the Cartan case on typical pairs; this general
definition also supplies channels for atypical pairs.

The map $\mathcal{C}_{\mu\to\nu}$ is completely positive because $J_\omega$ is proportional to a projector. Covariance, Schur's lemma, and trace normalization give $\Tr_\nu(\Pi_\omega)=(d_\omega/d_\mu)\id_{\mu^*}$. Hence
\[
    \Tr_\nu(J_\omega)
    =\frac{d_\mu}{d_\omega}\Tr_\nu(\Pi_\omega)
    =\id_{\mu^*}\ ,
\]
so the map is trace-preserving.

For symmetric powers this recovers Werner's cloner. For general irreps,
the map acts on representation states; the next formula explains its
cloning interpretation.

\subsection{Properties of the generalized cloner}\label{sec:properties}
\subsubsection{Cartan intertwiner formula}
In the Cartan case $\nu=\mu+\omega$, an intertwining isometry gives a representation that makes the connection to Werner's cloning map explicit~\eqref{eq:werner}.

\begin{prop}[Cartan intertwiner formula]\label{prop:choi}
Let $\mu,\nu,\omega$ be dominant integral highest weights with
$\nu=\mu+\omega$, and let
$V:\mathcal H_\nu\to\mathcal H_\mu\otimes\mathcal H_\omega$ be the isometric embedding into the Cartan component, characterized up to phase by
\begin{equation}
    VU_\nu(g)=(U_\mu(g)\otimes U_\omega(g))V
    \quad\text{for all }g\in\mathrm U(d)\ .
\end{equation}
Then the covariant channel
\begin{equation}
    \mathcal C_{\mu\to\nu}(\sigma_\mu)
    =\frac{d_\mu}{d_\nu}V^\dagger(\sigma_\mu\otimes\id_\omega)V
\end{equation}
has Choi operator $J=(d_\mu/d_\omega)\Pi_\omega$, where $\Pi_\omega$ projects onto the multiplicity-one PRV component $(\nu-\mu)^+$ in $\mu^*\otimes\nu$.
\end{prop}
\begin{proof}See Appendix~\ref{app:proofs}.
\end{proof}

To interpret the channel, we factor the target state through the same
embedding and identify the auxiliary state replaced by the maximally
mixed state.

Complexifying the Lie-algebra intertwining relation in
Proposition~\ref{prop:choi} extends it from $\mathrm U(d)$ to
$\mathrm{GL}(d,\mathbb C)$:
\[
V\pi_\nu(A)=\bigl(\pi_\mu(A)\otimes\pi_\omega(A)\bigr)V
\qquad A\in\mathrm{GL}(d,\mathbb C)\ .
\]
When all three highest weights are partitions, the representation
matrices are polynomial and this identity extends to all
$A\in M_d(\mathbb C)$.

For the state comparison between $\map C_{\mu\to\nu}(\rho_\mu)$ and $\rho_\nu$, take $\mu,\nu$ to be partitions of length at
most $r=\operatorname{rank}\rho$ and $\omega=\nu-\mu$ to be dominant.
Define
$s_\omega(x):=\Tr\pi_\omega(\rho)$ and
$\rho_\omega:=\pi_\omega(\rho)/s_\omega(x)$; $s_\omega$ is a Schur
polynomial when $\omega$ is a partition and a Schur Laurent polynomial
otherwise.
Taking $A=\rho$ and using $V^\dagger V=\id_\nu$ gives the exact factorization
\begin{equation}\label{eq:cartan_factorization}
\rho_\nu=\frac{s_\mu(x)\,s_\omega(x)}{s_\nu(x)}\,V^\dagger(\rho_\mu\otimes\rho_\omega)V\ ,
\end{equation}
to be compared with
\begin{equation}\label{eq:cloner_substitution}
\mathcal C_{\mu\to\nu}(\rho_\mu)=\frac{d_\mu\,d_\omega}{d_\nu}\,V^\dagger\!\Bigl(\rho_\mu\otimes\frac{\id_\omega}{d_\omega}\Bigr)V\ .
\end{equation}
The factor $\rho_\omega$ depends on the same unknown state as
$\rho_\mu$. The channel replaces it by the maximally mixed state,
the unique $\mathrm U(d)$-invariant state on $\h_\omega$, and projects
onto the Cartan component. For symmetric powers $\mu=(n)$ and
$\nu=(m)$ with $m\ge n$, this is Werner's construction:
$\rho_{(m)}\propto V^\dagger(\rho_{(n)}\otimes\rho_{(m-n)})V$,
with the auxiliary representation state replaced by the identity~\eqref{eq:werner}. For pure $\rho=\ketbra{\psi}{\psi}$ only symmetric powers carry weight
and $\rho_{(n)}=\ketbra{\psi}{\psi}^{\otimes n}$, giving literal cloning
of copies. 
Theorem~\ref{thm:main} bounds the resulting error in trace distance.

Beyond the resemblance to Werner's cloning map, one of us has recently used the generalized cloning map (for $\nu=\mu+\omega$) as a key construct in a protocol for mixed-states cloning~\cite{fei2026}, which is shown to be asymptotically optimal.  

Furthermore, Werner's cloner has also been applied in the quantum purity amplification for qubits~\cite{keyl2001rate}. It turns out a special case of generalized cloning map also arises in qudit quantum purity amplification~\cite{li2025purity,li2026_1,li2026_2}, the (asymptotically) optimal protocol is the generalized cloning map from $\mu$ to the fundamental irrep $\Box$, and its Choi matrix projects onto the PRV component of $\mu^*\otimes\Box$.

Hence, this generalization we propose is both formal and operational.

\subsubsection{Identity and qubit cases}
When $\mu=\nu$, the PRV component is trivial and the channel is the identity.

For $\mathrm{GL}(2)$, the irrep with highest weight $(a,b)$ is
$\det\nolimits^{\,b}\otimes\mathrm{Sym}^{a-b}(\mathbb C^2)$.
The determinant factor cancels from the conjugation action, so the
comparison of covariant channels reduces to symmetric powers.
\begin{rem}
    Let $\h_\mu=\mathrm{Sym}^n(\mathbb{C}^2)$ and $\h_\nu= \mathrm{Sym}^m(\mathbb{C}^2)$. For $m\ge n$, the generalized cloning map reduces to Werner's cloning map~\eqref{eq:werner}; for $m<n$, it is the partial-trace channel.
\end{rem}
Dualizing the input gives
$\mu^*\simeq\det\nolimits^{-n}\otimes\mathrm{Sym}^n(\mathbb C^2)$.
The Clebsch--Gordan decomposition therefore reads:
\begin{equation}
    \mu^*\otimes\nu \simeq
    \bigoplus_{k=0}^{\min(m,n)}
    \det\nolimits^{k-n}\otimes\mathrm{Sym}^{m+n-2k}(\mathbb C^2)\ .
\end{equation}
The PRV component is the term $k=\min(m,n)$, with highest weight
$(m-n,0)$ for $m\ge n$ and $(0,m-n)$ for $m<n$.
Its dimension, and hence the channel's Kraus rank, is $|m-n|+1$.
Upon restriction to $\mathrm{SU}(2)$, the decomposition can equivalently
be written $\bigoplus_{k=0}^{\min(m,n)}\mathrm{Sym}^{|m-n|+2k}(\mathbb C^2)$.
These channels are exactly those used by YCH in~\eqref{eq:qubitcloning}.

\subsubsection{Dominance order and dimension}
\begin{rem}[Minimal $\neq$ Smallest]\label{rem:prv_not_dim_min}
The PRV component in \(\mu^*\otimes \nu\) is not, in general, the irreducible summand of smallest dimension. This failure can occur even in the dominant case $\nu=\mu+\omega$. For example, let $\mu=(1,0,0)$ and $\nu=(5,1,0)$ be GL\((3)\) irreps. The PRV component is $\omega:=(\nu-\mu)^+ = \nu-\mu=(4,1,0)$, which has dimension $24$. However, $(5,0,0)$ is also a summand in the decomposition of $\mu^*\otimes \nu$\footnote{This is because $(1,0,0)\otimes (5,0,0)\cong(6,0,0)\oplus (5,1,0)$ from the Pieri rule.} and it is the smallest summand of dimension $21$.
\end{rem}

\subsubsection{Reverse cloner}
Reversing the representation pair gives the Petz recovery map with
respect to the maximally mixed input.
\begin{prop}[Reverse cloner]\label{prop:reverse}
Let $\C_{\mu\to\nu}:\mathcal B(\mathcal H_\mu)\to \mathcal B(\mathcal H_\nu)$
be the generalized cloning map, and let $\sigma_\mu := \id_\mu/d_\mu$.
Then the Petz recovery map of \(\C_{\mu\to\nu}\) with respect to \(\sigma_\mu\) is exactly the reverse  cloning map:
\[
\mathcal R_{\sigma_\mu,\C_{\mu\to\nu}}=\C_{\nu\to\mu}\ .
\]
Equivalently,
\[
\C_{\nu\to\mu}(X)=\frac{d_\nu}{d_\mu}\,\C_{\mu\to\nu}^{\dagger}(X),
\]
where \(\C_{\mu\to\nu}^{\dagger}\) is the adjoint map.
\end{prop}
\begin{proof}
See Appendix~\ref{app:proofs}.
\end{proof}
In the pure state case, the reverse cloner corresponds to tracing out the excess systems.

\subsubsection{Approximation of representation states}
We compare $\C_{\mu\to\nu}(\rho_\mu)$ with $\rho_\nu$, and similarly
in the reverse direction. The error is controlled by the size of the
dominant difference relative to the row gaps of $\mu$.
\begin{thm}[Cloning accuracy]\label{thm:main}
Let $\rho$ be a density matrix on $\mathbb C^d$, $d\ge2$, of rank $r\le d$ with
strictly decreasing positive eigenvalues $x_1>\cdots>x_r>0$.
Let $\mu,\nu$ be partitions with at most $r$ rows, and suppose
that $\omega:=\nu-\mu$ is a dominant integral weight. Set
\[
D:=\|\omega\|_1,\qquad
b_\mu:=\min_{1\le i\le\min(r,d-1)}(\mu_i-\mu_{i+1}).
\]
For the normalized irrep states $\rho_\mu,\rho_\nu$, the generalized
cloning channels satisfy, for a constant $C_{d,x}$ depending only on
the dimension and the nonzero spectrum,
\begin{align}\label{eq:cloning_trace_bound}
\frac12\|\C_{\mu\to\nu}(\rho_\mu)-\rho_\nu\|_1
&\le C_{d,x}\frac{D}{b_\mu+1},\nonumber\\
\frac12\|\C_{\nu\to\mu}(\rho_\nu)-\rho_\mu\|_1
&\le C_{d,x}\frac{D}{b_\mu+1}.
\end{align}
\end{thm}
The case $d=1$ is trivial, since both representation spaces are
one-dimensional. For $r<d$, the gap $b_\mu$ includes
$\mu_r-\mu_{r+1}=\mu_r$; for $r=d$, it involves only adjacent row differences.
In particular, $b_\mu=\Omega(n)$ gives errors $O(D/n)$, which vanish if $D=o(n)$. 
This trace-distance estimate is used in the achievability analysis in the following
section and will be proven in Section~\ref{sec:fidelity}.

Covariance also gives a useful structural property of the output.
\begin{prop}[Commutativity]\label{prop:commutativity}
Let $\rho$ have rank $r$, let $\mu,\nu$ be partitions with at most $r$
rows, and let $\mathcal N:\mathcal B(\mathcal H_\mu)\to
\mathcal B(\mathcal H_\nu)$ be any $\mathrm U(d)$-covariant channel.
Then
\[
[\mathcal N(\rho_\mu),\rho_\nu]=0.
\]
\end{prop}
\begin{proof}
In an eigenbasis of $\rho$, the state $\rho_\mu$ is invariant under
the diagonal unitary torus. By covariance, so is
$\mathcal N(\rho_\mu)$, which is therefore block diagonal in the
weight spaces of $\mathcal H_\nu$. Since $\rho_\nu$ is scalar on
each such space, the two operators commute.
\end{proof}

\section{Achievability}\label{sec:direct}

The qudit protocol maps all Schur blocks into one padded target
representation. For a rank-$r$ source with
$x_1>\cdots>x_r>0$, define the typical set by
\begin{equation}\label{eq:typical_set}
    \epsilon_n:=\sqrt{n/2}\log n,
    \qquad
    \mathcal T_{x,n}:=
    \left\{\lambda\vdash n:\ell(\lambda)\le r,\ 
    \max_{1\le i\le r}|\lambda_i-nx_i|\le\epsilon_n\right\}.
\end{equation}
Define the fixed target representation by
\begin{equation}\label{eq:target_rep}
    \Lambda_i:=
    \begin{cases}
        \left\lceil nx_i+(r-i+1)(2\epsilon_n+1)\right\rceil,
        &1\le i\le r,\\
        0,&r<i\le d.
    \end{cases}
\end{equation}
Thus $\Lambda$ is a partition of an integer slightly larger than $n$. The compression scheme is as follows.

\begin{tcolorbox}[colback=gray!5!white, colframe=gray!100!black, title=Qudit Compression Scheme]
\textbf{Encoding} \eqref{eq:optimal_encoder}:
\begin{enumerate}
    \item \textbf{Irrep Measurement:} Measure the Schur label $\lambda$ without disturbing the state within $\mathcal H_\lambda\otimes\mathcal M_\lambda$.
    \item \textbf{Cloning:} Discard the multiplicity register $\mathcal{M}_\lambda$ and apply $\mathcal{C}_{\lambda\to \Lambda}$, with the target in~\eqref{eq:target_rep}, to the remaining state $\rho_\lambda$. Store the output in a memory of dimension $\dim\mathcal H_{\Lambda}$.
\end{enumerate}

\textbf{Decoding} \eqref{eq:optimal_decoder}:
\begin{enumerate}
    \item \textbf{Irrep Sampling:} Pick a target representation $\lambda'$ at random from the known Schur--Weyl distribution $q_{\lambda',n}$.
    \item \textbf{Cloning:} Apply the generalized cloning channel $\mathcal{C}_{\Lambda\to \lambda'}$ to the state in the quantum memory.
    \item \textbf{State Reconstruction:} Append a new multiplicity register initialized in the maximally mixed state $\tau_{\lambda'}$, and perform the inverse Schur transform.
\end{enumerate}
\end{tcolorbox}

This protocol achieves the memory cost in Theorem~\ref{thm:qmdl}.
\begin{thm}[Achievability]\label{thm:achievability}
Let $\rho$ be a density matrix with rank $r\le d$ and strictly decreasing non-zero spectrum $x_1>\cdots>x_r>0$. Let $\rho_g:= U(g)\rho U(g)^\dagger$ for $g\in \mathrm U(d)$.
There exists a sequence of compression codes $(\mathcal E_n,\mathcal D_n)$ with quantum memory $M_n$ of size
\[
    |M_n|=L_{d,r}(n,x)+o(1),
\]
where $L_{d,r}(n,x)$ is defined in~\eqref{eq:result}, and error
\[
    \delta_n:= 
    \sup_{g\in \mathrm U(d)}
    \frac12\left\|
        \mathcal D_n\circ \mathcal E_n(\rho_g^{\otimes n}) - \rho_g^{\otimes n}
    \right\|_1
    =O\!\left(\frac{\log n}{\sqrt n}\right).
\]
The implicit constant depends only on $d,r$ and the fixed spectrum $x$.
\end{thm}

For pure states, we have
\begin{equation}\label{eq:achievable_pure}
    |M_n|=(d-1)\log n -\log (d-1)!+ o(1)\ .
\end{equation}

The memory size is $\log\dim\h_{\Lambda}$, which we evaluate
using the following consequence of the Weyl dimension formula.

\begin{lem}[Asymptotic Weyl dimension]\label{lem:weyl_asymptotic}
Fix $d$ and a spectrum with $x_1>\cdots>x_r>0$ and $x_i=0$ for
$i>r$. Let $\lambda$ be a partition, not necessarily of $n$, with
$\lambda_i=nx_i+O(\sqrt n\zeta_n)$ for $i\le r$ and
$\lambda_i=0$ for $i>r$, where $\zeta_n=o(\sqrt n)$. Then
\begin{equation}
\log\dim\h_\lambda
=L_{d,r}(n,x)+O\!\left(\frac{\zeta_n}{\sqrt n}+\frac1n\right).
\end{equation}
\end{lem}
\begin{proof}
The exact dimension is given by the Weyl dimension formula:
\begin{equation}\label{eq:weyl_dim}
    \dim \h_\lambda = \prod_{1 \le i < j \le d} \frac{\lambda_i - \lambda_j + j - i}{j - i} \ .
\end{equation}
For $i<j\le r$ the numerator is
$n(x_i-x_j)[1+O(\zeta_n/\sqrt n+1/n)]$, while for
$i\le r<j$ it is $nx_i[1+O(\zeta_n/\sqrt n+1/n)]$.
The factors with $r<i<j$ equal $j-i$ and cancel their denominators exactly. Since $d$ is fixed, multiplying the remaining factors gives
\begin{equation}
\dim\h_\lambda
=\frac{
n^{\frac{r(2d-r-1)}2}
\displaystyle\prod_{1\le i<j\le r}(x_i-x_j)
\displaystyle\prod_{i=1}^r x_i^{d-r}}
{\displaystyle\prod_{k=d-r}^{d-1}k!}
\left[1+O\!\left(\frac{\zeta_n}{\sqrt n}+\frac1n\right)\right],
\end{equation}
where the factorial ratio follows from
$\prod_{i<j}(j-i)=\prod_{k=1}^{d-1}k!$ (with the usual empty-product convention). Taking the logarithm proves the claim.
\end{proof}

\begin{proof}[Proof of Theorem~\ref{thm:achievability}]
For $d=1$, the source is a single known state and no memory is needed.
Assume henceforth that $d\ge2$.
We omit the subscript $g$ in $\rho_g$ in this proof. Every channel in the
protocol is $\mathrm U(d)$-covariant and the weights $q_{\lambda,n}$ are
independent of $g$, so the pointwise error is the same throughout the orbit.
By Schur--Weyl duality, since $\rho$ has rank $r$, the state
$\rho^{\otimes n}$ is supported on irreps with at most $r$ rows. Thus it is
unitarily equivalent to the block-diagonal state
\begin{equation}
    \rho^{\otimes n}\simeq\bigoplus_{\lambda : \ell(\lambda) \le r}
    q_{\lambda,n}\left(\rho_\lambda\otimes\tau_\lambda\right).
\end{equation}
Let $P_\lambda$ be the projection onto the block labeled by $\lambda$, and set $P:=\sum_{\lambda:\,\ell(\lambda)\le r}P_\lambda$, so every source state is supported in $P$.
The encoder is the channel on all of $(\mathbb C^d)^{\otimes n}$ defined by
\begin{equation}\label{eq:optimal_encoder}
    \map{E}_n(X)
    =\sum_{\lambda : \ell(\lambda) \le r}
    \map{C}_{\lambda\to \Lambda}\!\left(
    \Tr_{\map{M}_{\lambda}}(P_\lambda X P_\lambda)\right)
    +\Tr[(\id-P)X]\,
    \frac{\id_{\h_{\Lambda}}}{\dim\h_{\Lambda}},
\end{equation}
where the complement term makes $\map E_n$ completely positive and trace preserving on the full input space. It vanishes on every source state $\rho_g^{\otimes n}$, whose support lies in $P$.
Here $\map C_{\lambda\to\Lambda}$ is defined
by~\eqref{eq:gen_cloner}, with target~\eqref{eq:target_rep}.

This padding ensures that $\Lambda-\lambda$ is dominant for every
$\lambda\in\mathcal T_{x,n}$. Indeed, typicality and
$\lceil a\rceil-\lceil b\rceil\ge a-b-1$ give, for $1\le i<r$,
\begin{equation}
    \Lambda_i-\Lambda_{i+1}
    \ge n(x_i-x_{i+1})+2\epsilon_n
    \ge \lambda_i-\lambda_{i+1}.
\end{equation}
Moreover, $(\Lambda-\lambda)_r\ge\epsilon_n+1>0$.
Thus the rows of $\Lambda-\lambda$ are non-increasing and non-negative.

The corresponding decoder maps back to the original blocks:
\begin{equation}\label{eq:optimal_decoder}
    \map{D}_n(\cdot) = \bigoplus_{\lambda : \ell(\lambda) \le r}
    q_{\lambda,n}\left(\map{C}_{\Lambda\to \lambda}(\cdot)
    \otimes\tau_\lambda\right).
\end{equation}
The decoder samples $\lambda'$ independently of the measured $\lambda$,
so the composition contains cross terms with $\lambda\ne\lambda'$.
We bound its error through the reference state $\rho_{\Lambda}$.
The encoded memory state is
\begin{equation}
    \map E_n(\rho^{\otimes n})=\sum_{\lambda:\,\ell(\lambda)\le r} q_{\lambda,n}\,\map C_{\lambda\to\Lambda}(\rho_\lambda)\ .
\end{equation}
By the triangle inequality and the contractivity of the trace distance under $\map D_n$,
\begin{equation}
\begin{aligned}
    \delta_n &\le \frac12\left\|\map D_n\circ\map E_n(\rho^{\otimes n})-\map D_n(\rho_{\Lambda})\right\|_1 + \frac12\left\|\map D_n(\rho_{\Lambda})-\rho^{\otimes n}\right\|_1\\
    &\le \frac12\left\|\map E_n(\rho^{\otimes n})-\rho_{\Lambda}\right\|_1 + \sum_{\lambda'}q_{\lambda',n}\,\frac12\left\|\map C_{\Lambda\to\lambda'}(\rho_{\Lambda})-\rho_{\lambda'}\right\|_1\ ,
\end{aligned}
\end{equation}
where we used the block-diagonal identity
\[
\map D_n(\rho_{\Lambda})-\rho^{\otimes n}
=\bigoplus_{\lambda'}q_{\lambda',n}
\left(\map C_{\Lambda\to\lambda'}(\rho_{\Lambda})
-\rho_{\lambda'}\right)\otimes\tau_{\lambda'}
\]
and the multiplicativity of the trace norm under tensoring with $\tau_{\lambda'}$. Bounding the first term by convexity of the trace norm,
\begin{equation}
    \frac12\left\|\map E_n(\rho^{\otimes n})-\rho_{\Lambda}\right\|_1\le \sum_{\lambda}q_{\lambda,n}\,\frac12\left\|\map C_{\lambda\to\Lambda}(\rho_\lambda)-\rho_{\Lambda}\right\|_1\ ,
\end{equation}
we arrive at
\begin{equation}\label{eq:error_triangle}
    \delta_n\;\le\;\sum_{\lambda}q_{\lambda,n}\,\frac12\left\|\map C_{\lambda\to \Lambda}(\rho_\lambda)-\rho_{\Lambda}\right\|_1
    +\sum_{\lambda'}q_{\lambda',n}\,\frac12\left\|\map C_{\Lambda\to \lambda'}(\rho_{\Lambda})-\rho_{\lambda'}\right\|_1\ .
\end{equation}
Each of the two sums in~\eqref{eq:error_triangle} splits into typical and atypical contributions:
\begin{multline}\label{eq:error_twoterms}
\delta_n\le \max_{\lambda\in\mathcal{T}_{x,n}} \frac12\left\|\map{C}_{\lambda\to \Lambda}(\rho_\lambda)-\rho_{\Lambda}\right\|_1
+ \max_{\lambda'\in\mathcal{T}_{x,n}} \frac12\left\|\map{C}_{\Lambda\to \lambda'}(\rho_{\Lambda})-\rho_{\lambda'}\right\|_1
+2\sum_{\lambda\notin\mathcal T_{x,n}}q_{\lambda,n}\ ,
\end{multline}
having used the maximal trace distance bound $\frac12\|\rho-\sigma\|_1\le 1$ for the atypical terms in each sum; this is where the factor of $2$ on the tail term comes from.

The tail mass is controlled by Sanov's theorem via the following lemma:
\begin{lem}[Sanov's theorem~\cite{yang2016efficient,christandl2006spectra,o2021quantum,hayashi2017group}]\label{lem:tail_prob}
    Let $\lambda/n$ be the empirical probability distribution associated with a Young diagram distributed according to $q_{\lambda,n}$.
    Define the total variation distance by
    \[
        d(\lambda/n,x):=\frac12\sum_{i=1}^d|\lambda_i/n-x_i|.
    \]
    Then, for every $t>0$,
    \[
        \mathsf{Pr}\!\left[\lambda:d(\lambda/n,x)>t\right]
        \le (n+1)^{r(r+1)/2}e^{-2nt^2}.
    \]
\end{lem}
Lemma~\ref{lem:tail_prob} follows from the pointwise bound $\Pr[\lambda]\le(n+1)^{r(r-1)/2}2^{-nD(\lambda/n\|x)}$ on the supported Schur--Weyl diagrams (see~\cite{christandl2006spectra}). Summing over at most $(n+1)^r$ diagrams and using Pinsker's inequality $D(\lambda/n\|x)\ge(2/\ln2)\,d(\lambda/n,x)^2$ gives the claimed event bound.
The distribution $q_{\lambda,n}$ is supported on diagrams with at most $r$ rows. Hence, if $\lambda$ lies in its support but not in $\mathcal{T}_{x,n}$, then for some $i$,
$|\lambda_i/n-x_i|>(\log n)/\sqrt{2n}$, and hence
$d(\lambda/n,x)>(\log n)/(2\sqrt{2n})$. Lemma~\ref{lem:tail_prob} therefore gives
\begin{equation}\label{eq:tail_prob_bound}
    \sum_{\lambda\notin\mathcal T_{x,n}}q_{\lambda,n}
    \le (n+1)^{r(r+1)/2} e^{-(\log n)^2/4},
\end{equation}
which is superpolynomially small.

For typical states $\lambda\in\mathcal{T}_{x,n}$, the two maxima in~\eqref{eq:error_twoterms} are the forward and reverse cloning errors of Theorem~\ref{thm:main}, applied with $\mu=\lambda$ and $\nu=\Lambda$. The dominance hypothesis holds by construction, and $b_\lambda=\Omega(n)$ uniformly on the typical set because the positive eigenvalues are distinct and, when $r<d$, $x_r>0$.

Since $\Lambda-\lambda$ is non-negative and $|\lambda|=n$, its distance from every typical $\lambda$ is
\begin{equation}
    \|\Lambda-\lambda\|_1
    =\sum_i(\Lambda_i-\lambda_i)
    =|\Lambda|-n
    =O(\epsilon_n)
    =O(\sqrt n\log n).
\end{equation}
Theorem~\ref{thm:main} gives the trace-distance estimate directly, uniformly
over typical $\lambda$, in both directions:
\begin{equation}
\max\left\{
\frac12\|\C_{\lambda\to\Lambda}(\rho_\lambda)-\rho_{\Lambda}\|_1,
\frac12\|\C_{\Lambda\to\lambda}(\rho_{\Lambda})-\rho_\lambda\|_1
\right\}
=O\!\left(\frac{\log n}{\sqrt n}\right).
\end{equation}
Thus each typical maximum in~\eqref{eq:error_twoterms} has this order, and~\eqref{eq:tail_prob_bound} yields
\begin{equation}
    \delta_n
    \le O\!\left(\frac{\log n}{\sqrt n}\right)
    +2(n+1)^{\frac{r(r+1)}2}e^{-(\log n)^2/4}
    =O\!\left(\frac{\log n}{\sqrt n}\right).
\end{equation}

The memory is $\h_{\Lambda}$. Since
$\Lambda_i=nx_i+O(\sqrt n\log n)$ for $i\le r$ and
$\Lambda_i=0$ for $i>r$, Lemma~\ref{lem:weyl_asymptotic}
with $\zeta_n=\log n$ gives
\[
|M_n|=\log\dim\h_{\Lambda}=L_{d,r}(n,x)+o(1).
\]
\end{proof}

\section{Converse}\label{sec:converse}

The Koashi--Imoto theorem characterizes exact
compression~\cite{koashi2001compressibility,koashi2001possible}.
For approximate recovery, we prove a memory bound for irreducible
orbits under a spectral-gap assumption. It applies to arbitrary
encoding and decoding channels.

\begin{prop}[Memory bound for irreducible orbits]\label{prop:compact_orbit_memory}
Let $U:G\to\mathcal U(\h)$ be a continuous irreducible unitary
representation of a compact group on a finite-dimensional space
$\h$, with normalized Haar measure $\dd g$. Let $\tau$ be a state
on $\h$ with a simple largest
eigenvalue $p_0$ and corresponding rank-one projector $P$. Set
\[
p_1:=\|(\id-P)\tau(\id-P)\|_\infty,\qquad
\gamma:=p_0-p_1>0,\qquad
\tau_g:=U(g)\tau U(g)^\dagger.
\]
For any finite-dimensional memory $M$ and encoding and decoding channels
$\E:\mathcal B(\h)\to\mathcal B(M)$ and
$\D:\mathcal B(M)\to\mathcal B(\h)$, define
\[
\overline\delta:=
\int_G\frac12\|\D\circ\E(\tau_g)-\tau_g\|_1\,\dd g.
\]
Then
\begin{equation}\label{eq:compact_orbit_memory}
\dim M\ge\dim\h\left(1-\frac{\overline\delta}{\gamma}\right).
\end{equation}
\end{prop}
\begin{proof}
Write $\Phi=\D\circ\E$ and $P_g=U(g)PU(g)^\dagger$.
Trace preservation makes $\E(P_g)$ a state, so
$\E(P_g)\le\id_M$. Combining this with
$\tau_g\le p_1\id+\gamma P_g$, positivity gives
\[
\Phi(\tau_g)\le p_1\Phi(\id)+\gamma\D(\id_M).
\]
Schur's lemma gives $\int_G P_g\,\dd g=\id/\dim\h$. Testing against
$P_g$, averaging, and using trace preservation therefore yields
\[
p_0-\overline\delta
\le\int_G\Tr[P_g\Phi(\tau_g)]\,\dd g
\le\frac{p_1\Tr\Phi(\id)+\gamma\Tr\D(\id_M)}{\dim\h}
=p_1+\gamma\frac{\dim M}{\dim\h}.
\]
Rearranging proves the claim.
\end{proof}

At zero error this recovers the Koashi--Imoto incompressibility bound
for these orbit families: $\dim M\ge\dim\h$. For the representation
states used here, the spectral gap is bounded below uniformly in the
partition.

\begin{lem}[Uniform spectral gap]\label{lem:sector_gap}
Let $\rho$ have strictly decreasing nonzero spectrum
$x_1>\cdots>x_r>0$. Set
\[
q_x:=\max_{1\le i<r}\frac{x_{i+1}}{x_i},\qquad
\gamma_x:=(1-q_x)\prod_{1\le i<j\le r}
\left(1-\frac{x_j}{x_i}\right)>0,
\]
with $q_x:=0$ for $r=1$. For every partition $\lambda$ with at most
$r$ rows, $\rho_\lambda$ has a simple largest eigenvalue and its gap
$\gamma$, defined as in Proposition~\ref{prop:compact_orbit_memory},
satisfies $\gamma\ge\gamma_x$.
\end{lem}
\begin{proof}
See Appendix~\ref{app:proofs}.
\end{proof}

We transfer an arbitrary compression code to the padded target
representation and apply these two bounds. Haar-average error is
sufficient, so the converse result also applies to the worst-case error
criterion~\eqref{eq:error}.

\begin{thm}[Converse]\label{thm:converse}
Let $\rho$ have rank $r\le d$ and strictly decreasing nonzero spectrum
$x_1>\cdots>x_r>0$, and write $\rho_g=U(g)\rho U(g)^\dagger$.
Suppose a sequence of compression codes $(\E_n,\D_n)$ with memory
$M_n$ has Haar-average error
\[
\overline\delta_n:=
\int_{\mathrm U(d)}\frac12
\|\D_n\circ\E_n(\rho_g^{\otimes n})-\rho_g^{\otimes n}\|_1\,\dd g
\longrightarrow0.
\]
Then
\[
|M_n|\ge L_{d,r}(n,x)+o(1),
\]
where $L_{d,r}(n,x)$ is defined in~\eqref{eq:result}.
\end{thm}
\begin{proof}
Let $\mathcal A_n$ and
$\mathcal B_n$ be the achievability encoder and decoder
in~\eqref{eq:optimal_encoder} and~\eqref{eq:optimal_decoder}.
Write $\tau_{g,n}:=U_{\Lambda}(g)\rho_{\Lambda}
U_{\Lambda}(g)^\dagger$. The forward and reverse
estimates in the achievability proof give
\[
e_n:=\sup_g\left[
\frac12\|\mathcal A_n(\rho_g^{\otimes n})-\tau_{g,n}\|_1
+\frac12\|\mathcal B_n(\tau_{g,n})-\rho_g^{\otimes n}\|_1
\right]
=O\!\left(\frac{\log n}{\sqrt n}\right).
\]
The channels $\E_n\circ\mathcal B_n$ and $\mathcal A_n\circ\D_n$
form a code for the target orbit using the same memory $M_n$.
For each $g$, insert the intermediate states
$\mathcal A_n\circ\D_n\circ\E_n(\rho_g^{\otimes n})$ and
$\mathcal A_n(\rho_g^{\otimes n})$. The triangle inequality gives
\begin{multline}
\frac12\|\mathcal A_n\circ\D_n\circ\E_n\circ\mathcal B_n(\tau_{g,n})
-\tau_{g,n}\|_1\\
\le\frac12\bigl\|\mathcal A_n\circ\D_n\circ\E_n
\bigl(\mathcal B_n(\tau_{g,n})-\rho_g^{\otimes n}\bigr)\bigr\|_1
+\frac12\bigl\|\mathcal A_n
\bigl(\D_n\circ\E_n(\rho_g^{\otimes n})-\rho_g^{\otimes n}\bigr)\bigr\|_1
+\frac12\|\mathcal A_n(\rho_g^{\otimes n})-\tau_{g,n}\|_1\\
\le\frac12\|\mathcal B_n(\tau_{g,n})-\rho_g^{\otimes n}\|_1
+\frac12\|\D_n\circ\E_n(\rho_g^{\otimes n})-\rho_g^{\otimes n}\|_1
+\frac12\|\mathcal A_n(\rho_g^{\otimes n})-\tau_{g,n}\|_1.
\end{multline}
The second inequality uses contractivity under the channels
$\mathcal A_n\circ\D_n\circ\E_n$ and $\mathcal A_n$ in the first
and second terms, respectively. The middle term in the final
bound averages to $\overline\delta_n$, while the first and last
terms sum to at most $e_n$ for every $g$. Haar averaging gives
\[
\int_{\mathrm U(d)}\frac12
\|\mathcal A_n\circ\D_n\circ\E_n\circ\mathcal B_n(\tau_{g,n})
-\tau_{g,n}\|_1\,\dd g
\le\overline\delta_n+e_n.
\]
Proposition~\ref{prop:compact_orbit_memory} and
Lemma~\ref{lem:sector_gap} therefore imply, for sufficiently large $n$,
\[
|M_n|\ge\log\dim\h_{\Lambda}
+\log\left(1-\frac{\overline\delta_n+e_n}{\gamma_x}\right)
=\log\dim\h_{\Lambda}+o(1).
\]
Finally, Lemma~\ref{lem:weyl_asymptotic} applied to
$\Lambda_i=nx_i+O(\sqrt n\log n)$ for $i\le r$ and
$\Lambda_i=0$ for $i>r$ gives
$\log\dim\h_{\Lambda}=L_{d,r}(n,x)+o(1)$.
\end{proof}

\section{Universal quantum data compression and QMDL}\label{sec:redundancy}

The QMDL memory cost also determines the extra average code length
required by an unknown eigenbasis in universal lossless compression.
For a fixed spectrum $x_1>\cdots>x_r>0$, let $R_n(x)$ be the smallest
worst-case excess in ideal average length above $nS(\rho)$ among
codes that work for every eigenbasis. We will show that
\begin{equation}\label{eq:orbit_redundancy}
R_n(x)=L_{d,r}(n,x)-s_{\mathrm{blk}}(x)+o(1),
\end{equation}
where $L_{d,r}(n,x)$ is the optimal QMDL cost from
Theorem~\ref{thm:qmdl}, and $s_{\mathrm{blk}}(x)$ is the finite,
nonnegative entropy in~\eqref{eq:block_entropy_limit}. Thus the
optimal ideal average length is $nS(\rho)+R_n(x)$.

To make the coding criterion precise, a normalized code state
$\sigma_n$, with support containing that of
$\rho^{\otimes n}$, defines the ideal mean length
$-\Tr\rho^{\otimes n}\log\sigma_n$. Its excess over $nS(\rho)$ is
$D(\rho^{\otimes n}\|\sigma_n)$. The quantum Kraft inequality gives a
prefix code with lengths $\lceil-\log s_j\rceil$ in an eigenbasis of
$\sigma_n$, where $s_j$ are its nonzero eigenvalues
\cite{schumacher2001indeterminate,bostroem2002lossless,hayashi2010universal}.

For known $x$, the code may depend on the spectrum. The
Haar-averaged source state is
\[
\bar\rho_n:=\int_{\mathrm U(d)}\rho_g^{\otimes n}\,\dd g.
\]
More generally, a distribution $P_n$ on partitions $\lambda\vdash n$
with at most $d$ rows defines the invariant code state
\[
\sigma_{P_n,n}:=\bigoplus_\lambda P_n(\lambda)
\frac{\id_{\h_\lambda}}{\dim\h_\lambda}\otimes\tau_\lambda.
\]
For brevity, write $q(\lambda):=q_{\lambda,n}$ for the source's
Schur-label distribution. Schur--Weyl duality gives
$\bar\rho_n=\sigma_{q,n}$.
The next proposition identifies the optimal code state and
separates the cost of an unknown eigenbasis from the cost of
choosing a different label distribution.

\begin{prop}[Optimal overhead for an unknown eigenbasis]\label{prop:orbit_redundancy}
For any density matrix $\rho$ and any $n$, the Haar average
$\bar\rho_n$ is an exact minimax code state:
\begin{equation}\label{eq:orbit_minimax}
R_n(x):=\inf_{\sigma_n}\sup_{g\in\mathrm U(d)}
D(\rho_g^{\otimes n}\|\sigma_n)
=D(\rho^{\otimes n}\|\bar\rho_n)
=\mathbb E_\lambda[\log\dim\h_\lambda-S(\rho_\lambda)].
\end{equation}
For every invariant code state $\sigma_{P_n,n}$,
\begin{equation}\label{eq:invariant_split}
D(\rho^{\otimes n}\|\sigma_{P_n,n})
=R_n(x)+D(q\|P_n),
\end{equation}
with value $+\infty$ if the support condition fails.
\end{prop}
\begin{proof}
For any state $\sigma_n$, Haar averaging gives
\[
\int D(\rho_g^{\otimes n}\|\sigma_n)\,\dd g
=D(\bar\rho_n\|\sigma_n)+S(\bar\rho_n)-nS(\rho)
=D(\bar\rho_n\|\sigma_n)
+D(\rho^{\otimes n}\|\bar\rho_n).
\]
Thus the worst-case redundancy is at least
$D(\rho^{\otimes n}\|\bar\rho_n)$. The invariant state $\bar\rho_n$
attains this value for every $g$. Finally, the Schur--Weyl block
decomposition gives
\[
D(\rho^{\otimes n}\|\sigma_{P_n,n})
=D(q\|P_n)+\sum_\lambda q(\lambda)
[\log\dim\h_\lambda-S(\rho_\lambda)]=D(q\|P_n)+R_n(x) .
\]
\end{proof}

Haar averaging replaces each $\rho_\lambda$ by the maximally mixed
state on $\h_\lambda$. Its entropy therefore increases by
$\log\dim\h_\lambda-S(\rho_\lambda)$, the quantity averaged
in~\eqref{eq:orbit_minimax}. The average block entropy is already
included in the source entropy $nS(\rho)$: redundancy counts only
the additional uncertainty caused by the unknown eigenbasis.
QMDL instead measures the memory dimension needed to recover the
orbit, whose optimal logarithm is $L_{d,r}(n,x)+o(1)$.
The offset $s_{\mathrm{blk}}(x)$ is therefore the intrinsic mixedness
of a typical representation state. The following lemma evaluates
its limit and controls atypical blocks.

\begin{lem}[Block entropy]\label{lem:block_entropy}
Let $\rho$ have strictly decreasing nonzero spectrum
$x_1>\cdots>x_r>0$. The entropies $S(\rho_\lambda)$ are bounded
uniformly over all partitions $\lambda$ with at most $r$ rows.
For $r\ge2$, as
$g_\lambda:=\min_{i<r}(\lambda_i-\lambda_{i+1})\to\infty$,
\[
S(\rho_\lambda)\longrightarrow s_{\mathrm{blk}}(x),
\]
where
\begin{equation}\label{eq:block_entropy_limit}
s_{\mathrm{blk}}(x):=\sum_{1\le i<j\le r}
\left[-\log\left(1-\frac{x_j}{x_i}\right)
-\frac{x_j}{x_i-x_j}\log\frac{x_j}{x_i}\right].
\end{equation}
For $r=1$, every block entropy and the empty sum are zero.
\end{lem}
\begin{proof}
See Appendix~\ref{app:proofs}.
\end{proof}

Each summand in~\eqref{eq:block_entropy_limit} is the entropy of a
geometric distribution with ratio $x_j/x_i$.

On the typical Schur diagrams, the supported row gaps grow linearly
with $n$ when $r\ge2$. Lemmas~\ref{lem:block_entropy}
and~\ref{lem:weyl_asymptotic} therefore give uniform limits for
the block entropy and log dimension. The atypical contribution
is $o(1)$ by~\eqref{eq:tail_prob_bound}, the uniform entropy bound,
and $\log\dim\h_\lambda\le\binom d2\log(n+1)$. Consequently,
\[
\mathbb E_\lambda[S(\rho_\lambda)]=s_{\mathrm{blk}}(x)+o(1),
\qquad
\mathbb E_\lambda[\log\dim\h_\lambda]=L_{d,r}(n,x)+o(1).
\]
Substitution in~\eqref{eq:orbit_minimax} proves the claimed
relation~\eqref{eq:orbit_redundancy} between QMDL and
coding overhead.

When the spectrum is also unknown, the exact
split~\eqref{eq:invariant_split} adds the classical penalty
$D(q\|P_n)$. Appendix~\ref{app:unknown_spectrum} evaluates this
additional cost for Hayashi's universal codes.

\section{Trace-distance bounds for generalized cloning}\label{sec:fidelity}

We now prove Theorem~\ref{thm:main}. We retain the highest-weight Kraus
branch of each channel and use the Casimir gap to bound its deficit
in each weight block. The mean depth is bounded uniformly over all
representations, so averaging these deficits gives a finite error
bound without truncating the states. The remaining factors are a
scalar eigenvalue ratio and the ratio of Weyl dimensions.
Figure~\ref{fig:cloning_proof} shows how these estimates combine
for the forward channel, using the weight projectors $P_\delta,Q_\delta$
defined in Section~\ref{sec:prelim}.
By covariance it suffices to take $\rho=\operatorname{diag}
(x_1,\ldots,x_r,0,\ldots,0)$ throughout this section.

\subsection{Technical preliminaries}\label{sec:prelim}

We retain the assumptions and notation $\mu,\nu,\omega,D,b_\mu$ of
Theorem~\ref{thm:main}; in particular, $\nu=\mu+\omega$ and
$D=\|\omega\|_1=\sum_i|\omega_i|$. Recall that
$d_\lambda=\dim\mathcal H_\lambda$. For operators, $\|\cdot\|$ denotes
the operator norm and $\|\cdot\|_1$ the trace norm;
$X\ge Y$ means that $X-Y$ is positive semidefinite. For Hermitian
$X$, its positive part $X_+$ is obtained by replacing each negative
eigenvalue by zero.

We use the standard basis $e_i$ and the Euclidean
inner product on weights~\cite{fulton1991representation}.
The positive roots of the first $s$ coordinates are
$\Phi_s^+:=\{e_i-e_j:1\le i<j\le s\}$, and the simple roots are
$\alpha_i=e_i-e_{i+1}$. The Weyl group $S_d$ acts by permuting
coordinates. Write
\[
Q_+:=\sum_{i=1}^{d-1}\mathbb Z_{\ge0}\alpha_i,\qquad
Q_+^{(r)}:=\sum_{i=1}^{r-1}\mathbb Z_{\ge0}\alpha_i.
\]
For $\delta=\sum_i c_i\alpha_i$, its depth is $|\delta|:=\sum_i c_i$.
Thus $|e_i-e_j|=j-i$ for $i<j$. This depth is distinct from the
coordinate $\ell^1$ norm used in $D$. Here $\delta$ denotes a weight
offset; the compression error is denoted by $\delta_n$.

Every weight of $\mathcal H_\lambda$ is of the form
$\lambda-\delta$ with $\delta\in Q_+$. Let $\Pi_{\lambda-\delta}$
project onto this weight space, with multiplicity
$m_\lambda(\delta):=\operatorname{rank}\Pi_{\lambda-\delta}$.
We take the multiplicity and projector to be zero when the weight
is absent. The highest-weight space is one-dimensional; write
$\ket{\lambda}$ for a unit vector spanning it.

The Cartan generator $H_i$ and root generator $E_{e_i-e_j}$ are the
represented matrix units $E_{ii}$ and $E_{ij}$, respectively.
On a vector of weight $w$, $H_i$ acts by $w_i$, while $E_\alpha$
maps it to weight $w+\alpha$ or to zero. In particular,
$E_\alpha\ket{\lambda}=0$ for $\alpha>0$.
We use the invariant inner product for which $H_i^\dagger=H_i$
and $E_\alpha^\dagger=E_{-\alpha}$.

Let $\mathsf P_s(\delta)$ count the expressions of $\delta$ as a
nonnegative integer combination of roots in $\Phi_s^+$, with value
zero when no such expression exists. This is Kostant's partition
function. Ordered products of lowering operators applied to
$\ket{\lambda}$ span the representation, giving the Verma-module bound
$m_\lambda(\delta)\le\mathsf P_d(\delta)$.
For $\delta\in Q_+^{(r)}$, positive roots outside the
first $r$ coordinates cannot occur in a sum with all other
simple-root coefficients zero. Hence
\begin{equation}\label{eq:verma_multiplicity_bound}
m_\lambda(\delta)\le\mathsf P_d(\delta)=\mathsf P_r(\delta),
\qquad \delta\in Q_+^{(r)}.
\end{equation}

\begin{lem}[Multiplicity monotonicity]\label{lem:kostka}
Let $\mu,\omega$ be dominant integral highest weights of
$\mathrm{GL}(d,\mathbb C)$. Then
$m_\mu(\delta)\le m_{\mu+\omega}(\delta)$ for every $\delta\in Q_+$.
\end{lem}
\begin{proof}
Weight multiplicities count integral Gelfand--Tsetlin patterns
$(m_{i,j})_{1\le j\le i\le d}$, with top row $m_{d,j}=\mu_j$ and
interlacing inequalities
$m_{i+1,j}\ge m_{i,j}\ge m_{i+1,j+1}$.
Their weight is the vector of successive row-sum differences.
The map $m_{i,j}\mapsto m_{i,j}+\omega_j$ is injective, changes the
top row from $\mu$ to $\mu+\omega$, and shifts the weight by $\omega$.
It preserves interlacing: one inequality is unchanged, while the
other gains $\omega_j-\omega_{j+1}\ge0$. Thus it maps patterns counted
by $m_\mu(\delta)$ into those counted by $m_{\mu+\omega}(\delta)$.
\end{proof}

For $\rho=\operatorname{diag}(x_1,\ldots,x_r,0,\ldots,0)$ and a
partition $\lambda$ with at most $r$ rows,
\[
\rho_\lambda
=\sum_{\delta\in Q_+^{(r)}}
p_\lambda(\delta)\Pi_{\lambda-\delta},
\qquad
p_\lambda(\delta):=\frac{x^{\lambda-\delta}}{s_\lambda(x)},
\qquad x^w:=\prod_{i=1}^r x_i^{w_i}.
\]
Only these offsets contribute: both $\lambda$ and each supported weight
have zero coordinates after $r$, so their difference uses only
$\alpha_1,\ldots,\alpha_{r-1}$. The scalar $p_\lambda(\delta)$ is
defined even when the weight is absent; on each nonzero weight space
it is the eigenvalue of $\rho_\lambda$. The block probabilities are
$p_\lambda(\delta)m_\lambda(\delta)$ and sum to one. The normalizer
$s_\lambda(x)=\Tr\pi_\lambda(\rho)$ is the Schur polynomial
\cite{macdonald1995symmetric}. All sums over offsets below are over
$Q_+^{(r)}$, unless stated otherwise.

For $r\ge2$, set
\[
g_\lambda:=\min_{1\le i<r}(\lambda_i-\lambda_{i+1}).
\]
This measures the gaps within the first $r$ rows. In contrast,
$b_\mu$ from Theorem~\ref{thm:main} also includes the gap to the zero
row when $r<d$:
\[
b_\mu=\begin{cases}
\min\{g_\mu,\mu_r\},&2\le r<d,\\
g_\mu,&r=d.
\end{cases}
\]
Thus $b_\mu\le g_\mu$, while dominance of $\omega$ gives
$g_\nu\ge g_\mu$. For $r=1$, $b_\mu=\mu_1$, the only supported
offset is $\delta=0$, and no $g_\lambda$ is needed.

On $\mathcal H_\mu\otimes\mathcal H_\omega$, each Lie-algebra
generator acts as $X^{(\mu)}\otimes\id_\omega+
\id_\mu\otimes X^{(\omega)}$, so tensor-product weights add.
The vector $\ket{\mu}\otimes\ket{\omega}$ spans the one-dimensional
weight space of weight $\nu=\mu+\omega$ and generates the multiplicity-one
Cartan component. Choose the phase of the isometry $V$ in
Proposition~\ref{prop:choi} so that
$V\ket{\nu}=\ket{\mu}\otimes\ket{\omega}$.
Every other irreducible constituent has highest weight
$\chi=\nu-\eta$ with $0\ne\eta\in Q_+$.

For each supported offset $\delta$, define two projectors
on $\mathcal H_\mu\otimes\mathcal H_\omega$:
\begin{equation}\label{eq:proj_mu_omega}
P_\delta:=V\Pi_{\nu-\delta}V^\dagger,\qquad
Q_\delta:=\Pi_{\mu-\delta}\otimes\ketbra{\omega}{\omega}.
\end{equation}
The projector $P_\delta$ selects the embedded target weight space,
while $Q_\delta$ selects the input weight space with the auxiliary
system fixed to its highest weight. Both are supported on the full
tensor-product weight space $\mathcal W_{\nu-\delta}$ of weight
$\nu-\delta$. Since $V$ intertwines the Cartan action, $P_\delta$
is the restriction of $VV^\dagger$ to this space. For $Q_\delta$,
fixing the auxiliary weight to $\omega$ forces the input weight
to be $\mu-\delta$.

The quadratic Casimir uses the generators of the ambient
$\mathrm{GL}(d)$, even when the supported offsets lie in
$Q_+^{(r)}$:
\[
C_2=\sum_i H_i^2+
\sum_{\alpha>0}(E_{-\alpha}E_\alpha+E_\alpha E_{-\alpha}).
\]
It acts on $\mathcal H_\lambda$ by
$c_\lambda=(\lambda,\lambda+2\varrho)$, where
$\varrho=\frac12\sum_{\alpha\in\Phi_d^+}\alpha$ and
$(\varrho,\alpha_i)=1$. On the tensor product, the total Casimir
therefore acts by $c_\chi$ on each constituent $\mathcal H_\chi$.

\subsection{Proof}\label{sec:proof}

\begin{figure}[!t]
\centering
\begin{minipage}{\linewidth}
\centering\normalsize
\[
\begin{aligned}
\frac12\|\C_{\mu\to\nu}(\rho_\mu)-\rho_\nu\|_1
&\le\underbrace{\Tr\bigl(V\rho_\nu V^\dagger-\map K_\omega(\rho_\mu)\bigr)_+}_{
    \substack{\text{Retaining the highest-weight}\\
    \text{Kraus operator }\eqref{eq:kraus_proj}}}
\\[1em]
&=\sum_\delta\Tr\!\Bigl[
p_\nu(\delta)P_\delta
-\underbrace{p_\mu(\delta)}_{
    \substack{\ge p_\nu(\delta)\\\text{Lemma~\ref{lem:ratio}}}}
\map K_\omega(\Pi_{\mu-\delta})\Bigr]_+
\\[1em]
&\le\underbrace{\left(1-\frac{d_\mu}{d_\nu}\right)
+\frac{d_\mu}{d_\nu}\sum_\delta p_\nu(\delta)
\Tr\!\left[P_\delta(\id-Q_\delta)P_\delta\right]}_{
    \text{Insert the Kraus formula }\eqref{eq:kraus_proj}}
\\[1em]
&\le
{\color{teal!75!black}
\underbrace{\left(1-\frac{d_\mu}{d_\nu}\right)}_{
    \substack{\text{Dimension ratio}\\\text{Lemma~\ref{lem:dim_ratio}}}}}
+{\color{blue!65!black}
\underbrace{\frac{2D}{g_\mu+2}}_{
    \substack{\text{Projector estimate}\\
    \text{Lemmas~\ref{lem:shallow_multiplicities}--\ref{lem:perturbation}}}}}
\,{\color{violet!75!black}
\underbrace{\sum_\delta|\delta|\,p_\nu(\delta)m_\nu(\delta)}_{
    \substack{\text{Mean depth}\\\text{Lemma~\ref{lem:tail}}}}}
\\[1.5em]
&\le {\color{teal!75!black}\binom d2\frac{D}{b_\mu+1}}
+{\color{blue!65!black}\frac{2D}{g_\mu+2}}
\,{\color{violet!75!black}K(x)}
\le C_{d,x}\frac{D}{b_\mu+1}.
\end{aligned}
\]
\end{minipage}
\caption{Contributions to the forward trace-distance bound. The reverse-channel proof is analogous.}
\label{fig:cloning_proof}
\end{figure}

\emph{Stage 1: the highest-weight Kraus branch.}
Using the Cartan intertwiner
$V:\mathcal H_\nu\to\mathcal H_\mu\otimes\mathcal H_\omega$ from
Proposition~\ref{prop:choi}, embed the forward channel as
\begin{equation}\label{eq:embedded_channel}
    V\C_{\mu\to\nu}(X)V^\dagger
    =
    \frac{d_\mu}{d_\nu}P_\nu(X\otimes \id_\omega)P_\nu ,
\end{equation}
where \(P_\nu=VV^\dagger\in\mathcal B(\h_\mu\otimes\h_\omega)\) is the projector onto the embedded copy of \(\mathcal H_\nu\).
The embedded target state is
\begin{equation}\label{eq:embedded_target}
    \widetilde{\rho}_\nu
    :=
    V\rho_\nu V^\dagger
    =
    \sum_\delta p_\nu(\delta)P_\delta .
\end{equation}
Since the trace norm is invariant under a common isometric embedding,
\begin{equation}\label{eq:isometric_invariance_forward}
    \|\C_{\mu\to\nu}(\rho_\mu)-\rho_\nu\|_1
    =\|V\C_{\mu\to\nu}(\rho_\mu)V^\dagger-\widetilde{\rho}_\nu\|_1.
\end{equation}
The highest-weight Kraus operator is
\begin{equation}\label{eq:kraus_proj}
\mathcal K_\omega(X)
    :=
    \frac{d_\mu}{d_\nu}P_\nu(X\otimes \ketbra{\omega}{\omega})P_\nu .
\end{equation}

Since $\id_\omega\ge\ketbra{\omega}{\omega}$, positivity gives
\begin{equation}\label{eq:forward_monotonicity}
    V\C_{\mu\to\nu}(\rho_\mu)V^\dagger
    \ge\mathcal K_\omega(\rho_\mu)
    =\sum_\delta p_\mu(\delta)
       \mathcal K_\omega(\Pi_{\mu-\delta}).
\end{equation}

The omitted branches form a positive operator of trace
$1-\Tr\mathcal K_\omega(\rho_\mu)$. The triangle inequality and
isometric invariance therefore give
\begin{equation}\label{eq:branch_trace_comparison}
\frac12\|\C_{\mu\to\nu}(\rho_\mu)-\rho_\nu\|_1
\le\frac12\|\mathcal K_\omega(\rho_\mu)-\widetilde\rho_\nu\|_1
+\frac{1-\Tr\mathcal K_\omega(\rho_\mu)}2=\Tr(\widetilde\rho_\nu-\mathcal K_\omega(\rho_\mu))_+.
\end{equation}

Each $\mathcal K_\omega(\Pi_{\mu-\delta})$ is supported on the total-weight
sector $\nu-\delta$, so the branch can be compared with the target
block by block.

\emph{Stage 2: estimating the projectors.} We first identify the
dimensions of the shallow weight spaces, then use the spectral gap of
the total Casimir to compare their projectors.

\begin{lem}[Shallow weight multiplicities]\label{lem:shallow_multiplicities}
Let $\lambda$ be a dominant integral highest weight of
$\mathrm{GL}(d,\mathbb C)$ and fix $2\le r\le d$. Set
$g_\lambda:=\min_{1\le i<r}(\lambda_i-\lambda_{i+1})$.
For $\delta\in Q_+^{(r)}$ with $|\delta|\le g_\lambda$,
\[
m_\lambda(\delta)=\mathsf P_r(\delta),
\]
with $\mathsf P_r$ as defined in Section~\ref{sec:prelim}.
\end{lem}

\begin{proof}
See Appendix~\ref{app:proofs}.
\end{proof}

\begin{lem}[Highest-weight subspace perturbation]\label{lem:perturbation}
Let $\mu,\nu$ be partitions with at most $r$ rows and let
$\omega=\nu-\mu$ be a dominant integral highest weight of
$\mathrm{GL}(d,\mathbb C)$.
For $\delta\in Q_+^{(r)}$, let $P_\delta,Q_\delta$ be
the projectors in~\eqref{eq:proj_mu_omega}.
Set $D:=\|\omega\|_1$. For $r\ge2$, set
$g_\mu:=\min_{1\le i<r}(\mu_i-\mu_{i+1})$ and assume
$|\delta|\le g_\mu$. Then
\[
\operatorname{rank}P_\delta
=\operatorname{rank}Q_\delta
=m_\mu(\delta)
\]
and, for $r\ge2$,
\begin{equation}\label{eq:projector_casimir_bound}
    \|P_\delta-Q_\delta\|
    \le\sqrt{\frac{2(\delta,\omega)}{g_\mu+2}}
    \le\sqrt{\frac{2|\delta|D}{g_\mu+2}}\ .
\end{equation}
When $r=1$, only $\delta=0$ occurs and the two projectors coincide.
\end{lem}

\begin{proof}
If $|\delta|D=0$, the two projectors coincide.
Otherwise $r\ge2$ and $|\delta|\le g_\mu$ by hypothesis.
Work on the total-weight space $\mathcal W_{\nu-\delta}$.
Since $\omega$ is dominant, $g_\nu\ge g_\mu$. Lemma~\ref{lem:shallow_multiplicities}
therefore gives
\[
\operatorname{rank}P_\delta
=m_\nu(\delta)=\mathsf P_r(\delta)=m_\mu(\delta)
=\operatorname{rank}Q_\delta.
\]

Restrict the total quadratic Casimir $C_2$ to this weight space.
A constituent $\mathcal H_\chi$ of $\mathcal H_\mu\otimes\mathcal H_\omega$
that meets $\mathcal W_{\nu-\delta}$ has highest weight $\chi=\nu-\eta$ with
$\eta,\delta-\eta\in Q_+$. Thus $\eta$ is supported on the first
$r-1$ simple roots. For $\chi\ne\nu$, write
$\eta=\sum_{i<r}b_i\alpha_i$ with $\sum_i b_i=|\eta|\ge1$.
Dominance of $\chi$ and the supported row gaps of $\nu$ give
\begin{align}\label{eq:casimir_slice_gap}
c_\nu-c_\chi
&=(\nu+\chi+2\varrho,\eta)\nonumber\\
&=\sum_{i<r}b_i
  \bigl[(\nu_i-\nu_{i+1})+(\chi_i-\chi_{i+1})+2\bigr]\nonumber\\
&\ge(g_\mu+2)|\eta|\ge g_\mu+2.
\end{align}
The Cartan component occurs once, so
\[
c_\nu\id-C_2\ge(g_\mu+2)(\id-P_\delta).
\]

It remains to evaluate the Casimir deficit on $Q_\delta$. In the tensor-product
Casimir, the cross terms are
\[
2\sum_i H_i^{(\mu)}\otimes H_i^{(\omega)}
+2\sum_{\alpha>0}
\bigl(E_{-\alpha}^{(\mu)}\otimes E_\alpha^{(\omega)}
     +E_\alpha^{(\mu)}\otimes E_{-\alpha}^{(\omega)}\bigr).
\]
The root-exchange terms vanish when compressed by $Q_\delta$, because the
auxiliary register lies in its highest-weight line. The Cartan term is
$2(\mu-\delta,\omega)Q_\delta$. Using
$c_\nu=c_\mu+c_\omega+2(\mu,\omega)$, we obtain
\begin{equation}\label{eq:casimir_deficit_compression}
Q_\delta(c_\nu\id-C_2)Q_\delta=2(\delta,\omega)Q_\delta.
\end{equation}
Consequently,
\[
\|(\id-P_\delta)Q_\delta\|^2
=\|Q_\delta(\id-P_\delta)Q_\delta\|
\le\frac{2(\delta,\omega)}{g_\mu+2}.
\]
The projectors have equal rank, so
$\|P_\delta-Q_\delta\|=\|(\id-P_\delta)Q_\delta\|$.\footnote{For
projectors $P,Q$, the identity
$(P-Q)^2=P(\id-Q)P+(\id-P)Q(\id-P)$ is block diagonal with respect
to $\operatorname{ran}P\oplus\ker P$, and therefore
$\|P-Q\|=\max\{\|(\id-Q)P\|,\|(\id-P)Q\|\}$.
When the ranks agree, $PQP$ on $\operatorname{ran}P$ and $QPQ$ on
$\operatorname{ran}Q$ have the same eigenvalues: they are the squared singular
values of $PQ:\operatorname{ran}Q\to\operatorname{ran}P$.
Thus $P(\id-Q)P$ and $Q(\id-P)Q$ have equal norms, so the two
terms in the maximum coincide.}
Finally,
$0\le(\delta,\omega)\le|\delta|D$ gives~\eqref{eq:projector_casimir_bound}.
\end{proof}

\emph{Stage 3: averaging the block bounds.}
The uniformly bounded mean depth controls both the averaged
projector deficits and the tail mass needed to estimate the
eigenvalue ratio.

\begin{lem}[Uniform mean depth]\label{lem:tail}
Let \(\rho\) be a density matrix of rank \(r\le d\) with strictly decreasing nonzero spectrum \(x_1>\cdots>x_r>0\).
Let $\mu$ be a partition with at most $r$ rows. With $q_x$ as in
Lemma~\ref{lem:sector_gap}, including $q_x=0$ for $r=1$, set
\[
N:=|\Phi_r^+|=\binom r2,\qquad
K(x):=\frac{Nq_x}{(1-q_x)^{N+1}}.
\]
Then, uniformly over $\mu$,
\begin{equation}\label{eq:uniform_mean_depth}
\sum_{\delta}|\delta|\,p_\mu(\delta)m_\mu(\delta)
\le K(x).
\end{equation}
For $r=1$, both the mean depth and $K(x)$ are zero.
\end{lem}
\begin{proof}
For $r=1$, only the highest-weight line is supported, so the mean
depth vanishes. Assume $r\ge2$. Since
$s_\mu(x)\ge x^\mu$ and
$x^{-\delta}=\prod_{i<r}(x_{i+1}/x_i)^{c_i}$ for
$\delta=\sum_{i<r}c_i\alpha_i$,
\[
p_\mu(\delta)\le q_x^{|\delta|}.
\]
The multiplicity bound~\eqref{eq:verma_multiplicity_bound} gives
\[
\sum_{|\delta|=t}m_\mu(\delta)
\le[z^t]\prod_{\alpha\in\Phi_r^+}(1-z^{|\alpha|})^{-1}
\le[z^t](1-z)^{-N}
=\binom{t+N-1}{N-1}.
\]
Here $[z^t]$ denotes the coefficient of $z^t$.
Differentiating $(1-q_x)^{-N}$ now gives
\[
\sum_\delta|\delta|\,p_\mu(\delta)m_\mu(\delta)
\le\sum_{t\ge0}t\binom{t+N-1}{N-1}q_x^t
=\frac{Nq_x}{(1-q_x)^{N+1}}=K(x).
\]
\end{proof}

The eigenvalue ratio $p_\nu(\delta)/p_\mu(\delta)$ is independent
of $\delta$. Its deviation from one is controlled by the weight
sectors where the multiplicities can differ.

\begin{lem}[Eigenvalue ratio]\label{lem:ratio}
Let $\mu,\nu$ be partitions with at most $r$ rows and suppose
$\omega:=\nu-\mu$ is dominant. For a strictly decreasing positive
spectrum $x_1>\cdots>x_r>0$, set
\[
Z(x):=\frac{x^\omega s_\mu(x)}{s_\nu(x)}
=\frac{p_\nu(\delta)}{p_\mu(\delta)}.
\]
If $r=1$, then $Z(x)=1$. If $r\ge2$, then
\[
0\le1-Z(x)\le\frac{K(x)}{g_\mu+1}.
\]
\end{lem}
\begin{proof}
Normalization and Lemma~\ref{lem:kostka} give
\[
1-Z(x)
=\sum_\delta p_\nu(\delta)
\bigl[m_\nu(\delta)-m_\mu(\delta)\bigr]\ge0.
\]
For $r\ge2$, the multiplicities agree at depths $|\delta|\le g_\mu$
by Lemma~\ref{lem:shallow_multiplicities}, since $g_\nu\ge g_\mu$.
Since $0\le m_\nu(\delta)-m_\mu(\delta)\le m_\nu(\delta)$,
the remaining sum satisfies
\[
1-Z(x)
=\sum_{|\delta|>g_\mu}p_\nu(\delta)
\bigl[m_\nu(\delta)-m_\mu(\delta)\bigr]
\le\sum_{|\delta|>g_\mu}p_\nu(\delta)m_\nu(\delta).
\]
Both the depth and $g_\mu$ are integers, so $|\delta|>g_\mu$
implies $|\delta|\ge g_\mu+1$. Therefore
\begin{align*}
\sum_{|\delta|>g_\mu}p_\nu(\delta)m_\nu(\delta)
&\le\frac{1}{g_\mu+1}
\sum_{|\delta|>g_\mu}|\delta|\,p_\nu(\delta)m_\nu(\delta)\\
&\le\frac{1}{g_\mu+1}
\sum_\delta|\delta|\,p_\nu(\delta)m_\nu(\delta)
\le\frac{K(x)}{g_\mu+1},
\end{align*}
where the last step is Lemma~\ref{lem:tail},
Eq.~\eqref{eq:uniform_mean_depth}, applied to $\nu$.
For $r=1$, only $\delta=0$ occurs and both multiplicities are one.
\end{proof}

The other scalar factor is the ratio of representation dimensions.

\begin{lem}[Weyl dimension ratio]\label{lem:dim_ratio}
Under the partition and dominance assumptions of
Theorem~\ref{thm:main},
\[
0\le1-\frac{d_\mu}{d_\nu}
\le\sum_{i<j}
\frac{\omega_i-\omega_j}{\mu_i-\mu_j+j-i}
\le\binom d2\frac{D}{b_\mu+1}.
\]
\end{lem}
\begin{proof}
The Weyl dimension formula gives
\[
\frac{d_\mu}{d_\nu}
=\prod_{i<j}\left(1-
\frac{\omega_i-\omega_j}{\nu_i-\nu_j+j-i}\right).
\]
Each fraction lies in $[0,1)$. Using
$1-\prod_k(1-a_k)\le\sum_k a_k$ for $a_k\in[0,1]$ yields
\[
1-\frac{d_\mu}{d_\nu}
\le\sum_{i<j}\frac{\omega_i-\omega_j}{\nu_i-\nu_j+j-i}
\le\sum_{i<j}\frac{\omega_i-\omega_j}{\mu_i-\mu_j+j-i}.
\]
For $i>r$, the numerator is zero. For $i\le r$, it is at most $D$,
while the denominator is at least $b_\mu+1$.
\end{proof}

We are now ready to complete the proof of Theorem~\ref{thm:main}.

\begin{proof}[Proof of Theorem~\ref{thm:main}]
If $D=0$, then $\omega=0$ and both channels are the identity by
Propositions~\ref{prop:choi} and~\ref{prop:reverse}. Hence assume
$D\ge1$, as $D$ is an integer. We first bound the projector deficits
at every depth, then average them and compare the two channel outputs
with their target states. All sums below run over
$\delta\in Q_+^{(r)}$; absent weight spaces have zero
projector and multiplicity, as in Section~\ref{sec:prelim}.

\emph{Projector deficits and their averages.}
For $r\ge2$, set
\[
e_\delta:=\min\left\{1,\frac{2|\delta|D}{g_\mu+2}\right\}.
\]
At depths $|\delta|\le g_\mu$, Lemma~\ref{lem:perturbation}, together
with the general bound $\|P_\delta-Q_\delta\|\le1$ for
projectors, gives $\|P_\delta-Q_\delta\|^2\le e_\delta$.
At greater depths, integrality gives $|\delta|\ge g_\mu+1$, so
\[
\frac{2|\delta|D}{g_\mu+2}
\ge\frac{2(g_\mu+1)}{g_\mu+2}\ge1.
\]
Thus $e_\delta=1$ and the same projector bound holds without any
assumption on the ranks of the two weight spaces. For $r=1$, only
$\delta=0$ contributes; Lemma~\ref{lem:perturbation} gives
$P_0=Q_0$, and we set $e_0=0$.

The norm estimate controls the positive operator measuring the part
of $P_\delta$ outside $Q_\delta$:
\[
\|P_\delta(\id-Q_\delta)P_\delta\|
=\|(\id-Q_\delta)P_\delta\|^2
\le\|P_\delta-Q_\delta\|^2\le e_\delta.
\]
Since this operator is supported on $P_\delta$, it is at most
$e_\delta P_\delta$. Interchanging the projectors gives the analogous
bound on $Q_\delta$, and hence
\begin{equation}\label{eq:weight_block_bound}
P_\delta Q_\delta P_\delta\ge(1-e_\delta)P_\delta,\qquad
Q_\delta P_\delta Q_\delta\ge(1-e_\delta)Q_\delta.
\end{equation}
The cap $e_\delta\le1$ ensures that both lower bounds have
nonnegative coefficients. By~\eqref{eq:proj_mu_omega},
$\Tr P_\delta=m_\nu(\delta)$ and
$\Tr Q_\delta=m_\mu(\delta)$. For $r\ge2$, the corresponding
averaged trace deficits are controlled by
\begin{equation}\label{eq:averaged_projector_deficit}
M_\lambda:=\sum_\delta e_\delta p_\lambda(\delta)m_\lambda(\delta)
\le\frac{2D}{g_\mu+2}
\sum_\delta|\delta|\,p_\lambda(\delta)m_\lambda(\delta)
\le\frac{2D K(x)}{g_\mu+2},
\qquad \lambda\in\{\mu,\nu\},
\end{equation}
where the last inequality is Lemma~\ref{lem:tail},
Eq.~\eqref{eq:uniform_mean_depth}. Its constant is uniform over both
partitions. For $r=1$, define $M_\mu=M_\nu=0$ instead.

Equation~\eqref{eq:branch_trace_comparison} reduces the forward error
to the trace-distance term for the retained branch together with its
missing trace. We will control both terms at once by constructing
a positive remainder $R_{\mathrm f}$ such that
\[
\mathcal K_\omega(\rho_\mu)\ge\widetilde\rho_\nu-R_{\mathrm f}.
\]
The following elementary inequality shows that $\Tr R_{\mathrm f}$
then bounds the entire right-hand side of
\eqref{eq:branch_trace_comparison}. Let
$A\ge0$ have $\Tr A\le1$, let $B$ be a state, and suppose
$A\ge B-R$ for some $R\ge0$. Since
$\Tr(B-A)=1-\Tr A$, the positive-part identity and its variational
formula give
\begin{equation}\label{eq:positive_deficit_trace}
\frac{\|A-B\|_1+1-\Tr A}{2}
=\Tr(B-A)_+
=\max_{0\le L\le\id}\Tr[L(B-A)]\le\Tr R.
\end{equation}
Indeed, $B-A\le R$ implies
$\Tr[L(B-A)]\le\Tr(LR)\le\Tr R$ for every such $L$.
For normalized $A$, the left-hand side is the usual trace distance;
for a retained Kraus branch, it also includes the missing trace.

\emph{Forward channel.}
Lemma~\ref{lem:ratio} gives $p_\nu(\delta)=Z(x)p_\mu(\delta)$
with $0<Z(x)\le1$, so $p_\mu(\delta)\ge p_\nu(\delta)$.
Lemma~\ref{lem:dim_ratio} gives $0<d_\mu/d_\nu\le1$.
Since $Q_\delta$ lies in the total-weight sector $\nu-\delta$,
Eqs.~\eqref{eq:kraus_proj} and~\eqref{eq:proj_mu_omega} give
\[
\mathcal K_\omega(\Pi_{\mu-\delta})
=\frac{d_\mu}{d_\nu}P_\delta Q_\delta P_\delta.
\]
To construct the required remainder, use positivity of each term
and the eigenvalue comparison to obtain
\begin{align*}
\mathcal K_\omega(\rho_\mu)
=\frac{d_\mu}{d_\nu}\sum_\delta p_\mu(\delta)P_\delta Q_\delta P_\delta
\ge\frac{d_\mu}{d_\nu}\sum_\delta p_\nu(\delta)P_\delta Q_\delta P_\delta
=\widetilde\rho_\nu-R_{\mathrm f},
\end{align*}
where we used the target decomposition~\eqref{eq:embedded_target} and set
\[
R_{\mathrm f}:=\left(1-\frac{d_\mu}{d_\nu}\right)\widetilde\rho_\nu
+\frac{d_\mu}{d_\nu}\sum_\delta p_\nu(\delta)
P_\delta(\id-Q_\delta)P_\delta\ge0.
\]
These two terms account for the dimension ratio and the mismatch
between weight spaces, respectively. Taking traces and using
\eqref{eq:weight_block_bound} gives
\begin{align*}
\Tr R_{\mathrm f}
&=1-\frac{d_\mu}{d_\nu}
+\frac{d_\mu}{d_\nu}\sum_\delta p_\nu(\delta)
\Tr\bigl[P_\delta(\id-Q_\delta)P_\delta\bigr]\\
&\le1-\frac{d_\mu}{d_\nu}
+\frac{d_\mu}{d_\nu}\sum_\delta
 e_\delta p_\nu(\delta)m_\nu(\delta)
\le1-\frac{d_\mu}{d_\nu}+M_\nu.
\end{align*}
By~\eqref{eq:forward_monotonicity}, the retained branch has trace at
most one. We may therefore apply~\eqref{eq:positive_deficit_trace}
with $A=\mathcal K_\omega(\rho_\mu)$, $B=\widetilde\rho_\nu$,
and $R=R_{\mathrm f}$. The two terms on its left are exactly the
upper bound for the full channel in~\eqref{eq:branch_trace_comparison},
including the trace of the omitted branches. Consequently,
\begin{equation}\label{eq:forward_finite_deficit}
\frac12\|\C_{\mu\to\nu}(\rho_\mu)-\rho_\nu\|_1
\le\Tr R_{\mathrm f}
\le1-\frac{d_\mu}{d_\nu}+M_\nu.
\end{equation}

\emph{Reverse channel.}
Taking the adjoint in Proposition~\ref{prop:choi} and using
Proposition~\ref{prop:reverse} cancels the dimension factors:
\[
\C_{\nu\to\mu}(X)
=\frac{d_\nu}{d_\mu}\C_{\mu\to\nu}^{\dagger}(X)
=\Tr_\omega(VXV^\dagger).
\]
Define the isometry $Jv:=v\otimes\ket{\omega}$.
The partial trace is a sum of positive matrix elements over an
orthonormal auxiliary basis containing $\ket{\omega}$. Retaining
that one term gives
$\C_{\nu\to\mu}(\rho_\nu)\ge J^\dagger\widetilde\rho_\nu J$.
Moreover, in the total-weight sector $\nu-\delta$, fixing the auxiliary
weight to $\omega$ forces the input weight to be $\mu-\delta$.
Thus~\eqref{eq:proj_mu_omega} implies
\[
P_\delta J=P_\delta Q_\delta J,\qquad
J^\dagger Q_\delta J=\Pi_{\mu-\delta}.
\]
Using these identities, the decomposition~\eqref{eq:embedded_target},
and the second bound in~\eqref{eq:weight_block_bound}, we obtain
\begin{align*}
\C_{\nu\to\mu}(\rho_\nu)
\ge\sum_\delta p_\nu(\delta)
J^\dagger Q_\delta P_\delta Q_\delta J
\ge\sum_\delta p_\nu(\delta)(1-e_\delta)\Pi_{\mu-\delta}
=\sum_\delta Z(x)(1-e_\delta)
 p_\mu(\delta)\Pi_{\mu-\delta},
\end{align*}
where the last equality is Lemma~\ref{lem:ratio}.
Since $0<Z(x)\le1$ and $e_\delta\ge0$,
$Z(x)(1-e_\delta)\ge Z(x)-e_\delta$. The output is therefore at least
$\rho_\mu-R_{\mathrm r}$, where
\[
R_{\mathrm r}:=\sum_\delta(1-Z(x)+e_\delta)
 p_\mu(\delta)\Pi_{\mu-\delta}\ge0.
\]
The reverse output is a state, so applying~\eqref{eq:positive_deficit_trace}
and using $\sum_\delta p_\mu(\delta)m_\mu(\delta)=1$ yields
\begin{equation}\label{eq:reverse_finite_deficit}
\frac12\|\C_{\nu\to\mu}(\rho_\nu)-\rho_\mu\|_1
\le\Tr R_{\mathrm r}=1-Z(x)+M_\mu.
\end{equation}

\emph{Collecting the estimates.}
For $r\ge2$, substitute~\eqref{eq:averaged_projector_deficit} into
\eqref{eq:forward_finite_deficit}--\eqref{eq:reverse_finite_deficit}.
Lemma~\ref{lem:dim_ratio} bounds the forward dimension loss, while
Lemma~\ref{lem:ratio} bounds the reverse eigenvalue loss. Using
$b_\mu\le g_\mu$ and $D\ge1$ gives explicitly
\begin{align*}
\frac12\|\C_{\mu\to\nu}(\rho_\mu)-\rho_\nu\|_1
&\le\binom d2\frac{D}{b_\mu+1}
+\frac{2D K(x)}{g_\mu+2}
\le\left[\binom d2+2K(x)\right]\frac{D}{b_\mu+1},\\
\frac12\|\C_{\nu\to\mu}(\rho_\nu)-\rho_\mu\|_1
&\le\frac{K(x)}{g_\mu+1}
+\frac{2D K(x)}{g_\mu+2}
\le3K(x)\frac{D}{b_\mu+1}.
\end{align*}
For $r=1$, Lemma~\ref{lem:ratio} gives $Z(x)=1$, and
$M_\mu=M_\nu=0$. Equations~\eqref{eq:forward_finite_deficit} and
\eqref{eq:reverse_finite_deficit}, together with
Lemma~\ref{lem:dim_ratio}, then give the forward bound
$\binom d2 D/(b_\mu+1)$ and zero reverse error.
Thus in all cases one may take
$C_{d,x}=\binom d2+3K(x)$, with $K(x)$ as in
Lemma~\ref{lem:tail} (and $K(x)=0$ for $r=1$).
\end{proof}

\section*{Acknowledgments}
    We would like to thank Daniel Bump, Jiani Fei and Cynthia Yan for discussions and insights about covariant maps.
    YY acknowledges the support of the National Natural Science Foundation of China via the Excellent Young Scientists Fund (Hong Kong and Macau) Project 12322516 and the support of Research Grants Council (RGC) of Hong Kong SAR via the NSFC/RGC Joint Research Scheme N\_HKU7107/24. This material is based upon work supported by the U.S. Department of Energy, Office of Science, Office of High Energy and Nuclear Physics, under award number DE-SC0026551. The work of AM is supported by the Simons Foundation Grant No. 12574.
    
\section*{Generative AI usage}

The authors developed the connection between QMDL and free entropy
and the use of the PRV channel to extend YCH compression to qudits.
Generative AI assisted with proof exploration, mathematical revisions,
and language editing. Work with ChatGPT-o1, o3, and o4-mini developed
an initial strategy based on $\mathrm{GL}(d)$ Clebsch--Gordan
coefficients. Gemini 3 and 3.1 later suggested a projector-perturbation
argument using the Davis--Kahan theorem. The final proof of
Lemma~\ref{lem:perturbation} uses a direct Casimir argument suggested
by ChatGPT 6, yielding the finite bound in Theorem~\ref{thm:main}.
ChatGPT 6 and Fable 5.1 also assisted with manuscript review.
The authors reviewed and revised all AI-assisted material and take
responsibility for the mathematical arguments and the manuscript's content.

\appendix

\section{Proofs}\label{app:proofs}
\begin{proof}[Proof of Proposition~\ref{prop:choi}]
Let
\(
\ket{\Omega}=\sum_K \ket{\mu^*,K^*}\otimes\ket{\mu,K}
\)
be the unnormalized maximally entangled vector. By definition,
\(
J=(\mathrm{id}\otimes\mathcal C_{\mu\to\nu})(\ketbra{\Omega}{\Omega}),
\)
so using
\[
\mathcal C_{\mu\to\nu}(\sigma)
=\frac{d_\mu}{d_\nu}\,V^\dagger(\sigma\otimes\id_\omega)V,
\]
we get
\[
J=\frac{d_\mu}{d_\nu}\sum_{K,L}\sum_W
\ketbra{\mu^*,K^*}{\mu^*,L^*}
\otimes
\Bigl(
V^\dagger \ketbra{\mu,K;\omega,W}{\mu,L;\omega,W} V
\Bigr)\ .
\]

Define a linear map
\(
\tilde V:\mathcal H_{\mu^*}\otimes\mathcal H_\nu\to \mathcal H_\omega
\) by
\[
\tilde V\bigl(\ket{\mu^*,K^*}\otimes\ket{\nu,N}\bigr)
:=
(\bra{\mu,K}\otimes\id_\omega)V\ket{\nu,N}\ .
\]
Then its matrix elements satisfy
\[
\bra{\mu^*,K^*;\nu,N}\tilde V^\dagger \tilde V\ket{\mu^*,L^*;\nu,N'}
=
\sum_W
\bra{\nu,N}V^\dagger\ket{\mu,K;\omega,W}
\bra{\mu,L;\omega,W}V\ket{\nu,N'}.
\]
Comparing with the expression above for \(J\), we obtain
\[
J=\frac{d_\mu}{d_\nu}\,\tilde V^\dagger \tilde V.
\]

Contracting the $\mu$ factor of the intertwiner $V$ gives the intertwiner
\(
\tilde V:\mu^*\otimes\nu\to\omega.
\)
Hence \(\tilde V^\dagger\tilde V\) commutes with the representation \(\mu^*\otimes\nu\), and its support lies in the \(\omega\)-isotypic component.

The Cartan component has multiplicity one, so
\[
\dim\operatorname{Hom}_{U(d)}(\nu,\mu\otimes\omega)
=\dim\operatorname{Hom}_{U(d)}(\mu^*\otimes\nu,\omega)=1.
\]
Therefore \(\omega\) occurs with multiplicity one in \(\mu^*\otimes\nu\). If \(\Pi_\omega\) denotes the projector onto that unique copy, then $\tilde V^\dagger\tilde V$ is proportional to $\Pi_\omega$.
Its trace is
\[
\Tr(\tilde V^\dagger\tilde V)
=
\sum_{N,K}\left\|
(\bra{\mu,K}\otimes\id_\omega)V\ket{\nu,N}
\right\|^2
=
\sum_N \|V\ket{\nu,N}\|^2
=
\Tr(V^\dagger V)
=
d_\nu,
\]
since \(V\) is an isometry. Since $\Tr\Pi_\omega=d_\omega$, this gives
\[
\tilde V^\dagger\tilde V=\frac{d_\nu}{d_\omega}\Pi_\omega.
\]
Substituting back into the Choi formula gives
\[
J=\frac{d_\mu}{d_\nu}\tilde V^\dagger\tilde V
=\frac{d_\mu}{d_\omega}\Pi_\omega.
\]

Since $(\nu-\mu)^+=\omega$, this is the required PRV Choi operator.
\end{proof}

\begin{proof}[Proof of Proposition~\ref{prop:reverse}]
Covariance makes $\C_{\mu\to\nu}(\sigma_\mu)$ invariant.
Since $U_\nu$ is irreducible, Schur's lemma and trace preservation give
\[
\C_{\mu\to\nu}(\sigma_\mu)=\frac{\id_\nu}{d_\nu}.
\]

Therefore the Petz map with respect to \(\sigma_\mu\) is
\[
\mathcal R_{\sigma_\mu,\C_{\mu\to\nu}}(X)
=
\frac1{d_\mu}\,\C_{\mu\to\nu}^{\dagger}(d_\nu X)
=
\frac{d_\nu}{d_\mu}\,\C_{\mu\to\nu}^{\dagger}(X).
\]
So it remains to identify the scaled adjoint with the reverse cloning channel.

Let
\(
\operatorname{vec}:\operatorname{Hom}(\mathcal H_\mu,\mathcal H_\nu)\longrightarrow \mathcal H_\mu^*\otimes \mathcal H_\nu
\)
be the canonical vectorization isometry. Choose Hilbert--Schmidt
orthonormal operators $K_a\in\operatorname{Hom}(\mathcal H_\mu,\mathcal H_\nu)$,
$1\le a\le d_\omega$, whose vectorizations form an orthonormal basis
for the range of $\Pi_\omega$.
Then
\[
J_\omega
=
\frac{d_\mu}{d_\omega}
\sum_{a=1}^{d_\omega}
\ketbra{\operatorname{vec}(K_a)}{\operatorname{vec}(K_a)},
\]
and hence
\[
\C_{\mu\to\nu}(X)
=
\frac{d_\mu}{d_\omega}
\sum_{a=1}^{d_\omega}
K_a X K_a^\dagger.
\]
Consequently,
\[
\frac{d_\nu}{d_\mu}\,\C_{\mu\to\nu}^{\dagger}(Y)
=
\frac{d_\nu}{d_\omega}
\sum_{a=1}^{d_\omega}
K_a^\dagger Y K_a.
\]

Now the span
\(
\operatorname{span}\{K_a\}
\subset \operatorname{Hom}(\mathcal H_\mu,\mathcal H_\nu)
\)
carries the irrep \(\omega\) under the action
\[
K\longmapsto U_\nu(g) K U_\mu(g)^\dagger.
\]
Taking Hilbert--Schmidt adjoints defines an antiunitary intertwiner
\[
K\longmapsto K^\dagger
\]
from \(\operatorname{Hom}(\mathcal H_\mu,\mathcal H_\nu)\) to \(\operatorname{Hom}(\mathcal H_\nu,\mathcal H_\mu)\). Therefore
\(
\operatorname{span}\{K_a^\dagger\}
\subset \operatorname{Hom}(\mathcal H_\nu,\mathcal H_\mu)
\)
carries the dual irrep \(\omega^*\). Since
\[
\omega^*=(-w_0\omega)=(\mu-\nu)^+,
\]
this is exactly the PRV component for the reversed pair \((\nu,\mu)\).

The operators $K_a^\dagger$ remain Hilbert--Schmidt orthonormal,
so the scaled adjoint has Choi operator
$(d_\nu/d_{\omega^*})\Pi_{\omega^*}$, with
$d_{\omega^*}=d_\omega$. This is the defining Choi operator of
$\C_{\nu\to\mu}$, proving
\[
\frac{d_\nu}{d_\mu}\,\C_{\mu\to\nu}^{\dagger}
=
\C_{\nu\to\mu}.
\]
Combining this with the Petz formula above proves the claim.
\end{proof}

\begin{proof}[Proof of Lemma~\ref{lem:sector_gap}]
For $r=1$, $\rho_\lambda$ is pure and $\gamma=\gamma_x=1$.
Assume $r\ge2$ and use the weight notation of
Section~\ref{sec:prelim}. The nonzero spectrum agrees with that of
the corresponding $\mathrm{GL}(r)$ representation. The highest-weight
line has eigenvalue $p_0=x^\lambda/s_\lambda(x)$. Every other
nonzero eigenvalue, at weight $\lambda-\delta$, divided by $p_0$ is
\[
x^{-\delta}
=\prod_{i=1}^{r-1}\left(\frac{x_{i+1}}{x_i}\right)^{c_i}
\le q_x,
\qquad
\delta=\sum_{i=1}^{r-1}c_i\alpha_i,\quad
c_i\in\mathbb Z_{\ge0},\quad \sum_i c_i>0.
\]
Since $q_x<1$ and the highest-weight space is one-dimensional,
the largest eigenvalue is simple and $p_1\le q_xp_0$.

It remains to bound $p_0$ below uniformly in $\lambda$.
Expanding the Schur polynomial over weight spaces and applying
\eqref{eq:verma_multiplicity_bound} term by term gives
\[
p_0^{-1}=\frac{s_\lambda(x)}{x^\lambda}
=\sum_{\delta\in Q_+^{(r)}}m_\lambda(\delta)x^{-\delta}
\le\sum_{\delta\in Q_+^{(r)}}\mathsf P_r(\delta)x^{-\delta}.
\]
Here all terms are nonnegative. By definition,
$\mathsf P_r(\delta)$ counts the choices of nonnegative integers
$k_{ij}$ satisfying
\[
\delta=\sum_{1\le i<j\le r}k_{ij}(e_i-e_j).
\]
For each such choice, the definition $x^w=\prod_{\ell=1}^r x_\ell^{w_\ell}$
gives $x^{-(e_i-e_j)}=x_j/x_i$, and hence
\[
x^{-\delta}=\prod_{1\le i<j\le r}
\left(\frac{x_j}{x_i}\right)^{k_{ij}}.
\]
Thus summing over all offsets and their decompositions is equivalent
to summing independently over every $k_{ij}$, yielding
\[
\sum_{\delta\in Q_+^{(r)}}\mathsf P_r(\delta)x^{-\delta}
=\prod_{1\le i<j\le r}\sum_{k=0}^{\infty}
\left(\frac{x_j}{x_i}\right)^k
=\prod_{1\le i<j\le r}\left(1-\frac{x_j}{x_i}\right)^{-1}.
\]
The geometric series converge because $x_j/x_i<1$ for $i<j$.
Taking reciprocals gives
$p_0\ge\prod_{i<j\le r}(1-x_j/x_i)$, and hence
\[
\gamma=p_0-p_1\ge(1-q_x)p_0
\ge(1-q_x)\prod_{1\le i<j\le r}
\left(1-\frac{x_j}{x_i}\right)=\gamma_x.
\]
\end{proof}

\begin{proof}[Proof of Lemma~\ref{lem:block_entropy}]
Assume $r\ge2$ and use the weight notation of
Section~\ref{sec:prelim}. All sums below are over
$\delta\in Q_+^{(r)}$. Set
\[
Z_\lambda:=\sum_\delta m_\lambda(\delta)x^{-\delta},\qquad
B_\lambda:=\sum_\delta m_\lambda(\delta)x^{-\delta}\log x^\delta.
\]
The nonzero eigenvalues of $\rho_\lambda$ are $x^{-\delta}/Z_\lambda$,
with multiplicity $m_\lambda(\delta)$, so
\[
S(\rho_\lambda)=\log Z_\lambda+B_\lambda/Z_\lambda.
\]
The Verma-module bound~\eqref{eq:verma_multiplicity_bound} gives
$m_\lambda(\delta)\le\mathsf P_r(\delta)$, where
$\mathsf P_r$ is the positive-root partition function.
Replacing $m_\lambda$ by $\mathsf P_r$ defines the convergent sums
$Z_\infty,B_\infty$. Their product formulas are
\[
Z_\infty=\prod_{i<j\le r}\left(1-\frac{x_j}{x_i}\right)^{-1},
\qquad
\frac{B_\infty}{Z_\infty}
=\sum_{i<j\le r}\frac{x_j}{x_i-x_j}\log\frac{x_i}{x_j}.
\]
Both formulas follow by summing the independent nonnegative
occupation numbers of the positive roots. Since
$1\le Z_\lambda\le Z_\infty$ and $0\le B_\lambda\le B_\infty$,
the entropies are uniformly bounded.
By Lemma~\ref{lem:shallow_multiplicities},
$m_\lambda(\delta)=\mathsf P_r(\delta)$ whenever $|\delta|\le g_\lambda$.
The omitted tails of the two convergent positive sums tend to zero
as $g_\lambda\to\infty$. Hence $Z_\lambda\to Z_\infty$ and
$B_\lambda\to B_\infty$, giving~\eqref{eq:block_entropy_limit}.
The case $r=1$ is immediate because $\rho_\lambda$ is pure.
\end{proof}

\begin{proof}[Proof of Lemma~\ref{lem:shallow_multiplicities}]
We show that at these depths all nonidentity terms in Kostant's
multiplicity formula vanish, so the multiplicity is exactly the
number of positive-root partitions of $\delta$. The formula
\cite{kostant1959multiplicity} reads
\[
m_\lambda(\delta)
=\sum_{w\in S_d}\operatorname{sgn}(w)\,
\mathsf P_d\bigl(\delta-[\lambda+\varrho-w(\lambda+\varrho)]\bigr).
\]
Recall that $\mathsf P_d$ is zero outside $Q_+$: its argument must
have nonnegative coefficients in the simple-root basis.
Write $\delta=\sum_{i=1}^{d-1}c_i\alpha_i$. By hypothesis,
$c_i\ge0$, $c_i=0$ for $i\ge r$, and
$c_i\le|\delta|\le g_\lambda$ for every $i$.

Fix $w\ne\mathrm{id}$, and let $k\le d-1$ be the first coordinate
moved by $w$. Then $w^{-1}(i)=i$ for $i<k$, while $w^{-1}(k)>k$.
For a vector with zero coordinate sum, the coefficient of
$\alpha_k=e_k-e_{k+1}$ in its simple-root expansion is the sum
of its first $k$ coordinates. Thus the $\alpha_k$ coefficient of
$\lambda+\varrho-w(\lambda+\varrho)$ is
\begin{align*}
\sum_{i=1}^k\bigl[(\lambda+\varrho)_i
 -(w(\lambda+\varrho))_i\bigr]
&=(\lambda+\varrho)_k-(\lambda+\varrho)_{w^{-1}(k)}\\
&\ge(\lambda+\varrho)_k-(\lambda+\varrho)_{k+1}=\lambda_k-\lambda_{k+1}+1.
\end{align*}
Here dominance of $\lambda$ and
$\varrho_i-\varrho_{i+1}=1$ make $\lambda+\varrho$ strictly decreasing,
which justifies the inequality.
If $k<r$, the last expression is at least $g_\lambda+1>c_k$.
If $k\ge r$, it is at least one, whereas $c_k=0$.
In either case, subtracting this vector from $\delta$ produces a
negative $\alpha_k$ coefficient. The argument of $\mathsf P_d$
is therefore outside $Q_+$, and the term indexed by $w$ vanishes.
Only $w=\mathrm{id}$ remains, giving
$m_\lambda(\delta)=\mathsf P_d(\delta)$.

Finally, each positive root has the expansion
$e_i-e_j=\alpha_i+\cdots+\alpha_{j-1}$. A root with $j>r$
would contribute a positive coefficient to some $\alpha_k$ with
$k\ge r$. Such contributions cannot cancel in a nonnegative sum
of positive roots, while all these coefficients of $\delta$ are zero.
Thus every root partition of $\delta$ uses only roots in
$\Phi_r^+$, so $\mathsf P_d(\delta)=\mathsf P_r(\delta)$.
\end{proof}

\section{Additional cost of an unknown spectrum}\label{app:unknown_spectrum}

We use the notation of Section~\ref{sec:redundancy}. When the spectrum is unknown, the code's weights $P_n$ must be
chosen independently of $x$. Equation~\eqref{eq:invariant_split}
identifies the resulting penalty as the classical relative
entropy $D(q\|P_n)$; the eigenbasis contribution remains
$R_n(x)$. We now evaluate this additional cost for Hayashi's
universal codes.

Assume full rank and $x_1>\cdots>x_d>0$. Write\footnote{Our
$v(x)$ is the permutation expression in
\cite[Eq.~(10)]{hayashi2010universal}, following
\cite{matsumoto2007universal}:
\[
v(x)=\frac{1}{\Delta(x)}
\sum_{\pi\in S_d}\operatorname{sgn}(\pi)
\left(\prod_i x_i^{d-\pi(i)}\right)
\log\left(\prod_i x_i^{d-\pi(i)}\right).
\]
Indeed, applying $\sum_i x_i\log x_i\,\partial_{x_i}$ to the product
and determinant expansions of
$\Delta(x)=\det[x_i^{d-j}]_{i,j=1}^d$ gives the two forms.}
\[
\Delta(x):=\prod_{i<j}(x_i-x_j),\qquad
v(x):=\sum_{i<j}\frac{x_i\log x_i-x_j\log x_j}{x_i-x_j}.
\]
Equation~\eqref{eq:block_entropy_limit} gives
$s_{\mathrm{blk}}(x)=v(x)-\log\Delta(x)$, so
\begin{equation}\label{eq:orbit_redundancy_fullrank}
R_n(x)=\frac{d(d-1)}2\log n+2\log\Delta(x)-v(x)
-\sum_{k=1}^{d-1}\log k!+o(1).
\end{equation}

Hayashi's uniform code $\sigma_{U,n}=\sigma_{P_{U,n},n}$ assigns equal
weight to each of the $\sim n^{d-1}/[d!(d-1)!]$ Young diagrams
\cite{hayashi2010universal}. Define
\begin{equation}\label{eq:hayashi_C}
C(x):=2\log\Delta(x)-\frac12\sum_i\log x_i-v(x).
\end{equation}
With the Weyl denominator
$\prod_{i<j}(j-i)=\prod_{k=1}^{d-1}k!$, he obtained the expansion
\begin{align}\label{eq:hayashi_pointwise}
D(\rho^{\otimes n}\|\sigma_{U,n})
={}&\frac{d^2-1}{2}\log n+C(x)-\frac{d-1}{2}\log(2\pi e)\nonumber\\
&-\sum_{k=1}^{d-1}\log k!-\log[d!(d-1)!]+o(1).
\end{align}
Subtracting~\eqref{eq:orbit_redundancy_fullrank} in the exact
decomposition~\eqref{eq:invariant_split} gives
\begin{equation}\label{eq:spectral_redundancy_uniform}
D(q\|P_{U,n})
=\frac{d-1}{2}\log n-\frac12\sum_i\log x_i
-\frac{d-1}{2}\log(2\pi e)-\log[d!(d-1)!]+o(1).
\end{equation}
The extra leading cost is therefore $\frac{d-1}{2}\log n$, as in
classical coding with $d-1$ unknown spectral
parameters~\cite{clarke1990information}. Its constant contains
the multinomial Fisher determinant and the normalization of the
uniform spectral prior.

The uniform code is one choice of prior. Optimizing the asymptotic
constant over codes for all fixed full-rank states gives Hayashi's
minimax value~\cite{hayashi2010universal}:
\begin{align}\label{eq:hayashi_minimax}
&\inf_{\{\sigma_n\}}\sup_{\rho>0}\lim_{n\to\infty}
\left[D(\rho^{\otimes n}\|\sigma_n)-\frac{d^2-1}{2}\log n\right]\nonumber\\
&\qquad=-\frac{d-1}{2}\log(2\pi e)
-\sum_{k=1}^{d-1}\log k!+\log\int 2^{C(x)}\,\dd x.
\end{align}
The infimum is over code sequences for which the displayed limit
exists in the extended real numbers for every fixed full-rank state
$\rho$.
The integral is over ordered positive spectra with $\sum_i x_i=1$
and $\dd x=\dd x_1\cdots\dd x_{d-1}$. An asymptotically optimizing
code is the continuous mixture
\[
\sigma_{J,n}:=\int J(x)\,\overline{\rho(x)^{\otimes n}}\,\dd x,
\qquad
J(x):=\frac{2^{C(x)}}{\int2^{C(y)}\,\dd y},
\]
where $\rho(x)=\operatorname{diag}(x)$ and the overline denotes
Haar averaging. This is the quantum analogue of the Jeffreys-prior
construction~\cite{clarke1990information,clarke1994jeffreys}.
Changing the prior changes the classical term
in~\eqref{eq:invariant_split}. For both codes, the leading overhead is
\[
\frac{d^2-1}{2}\log n
=\underbrace{\frac{d(d-1)}2\log n}_{\text{unknown eigenbasis}}
+\underbrace{\frac{d-1}2\log n}_{\text{unknown spectrum}}.
\]
QMDL determines the first contribution through its additive
constant, with the block-entropy correction
in~\eqref{eq:orbit_redundancy}; universal compression must also
pay for the unknown spectrum.

\bibliographystyle{utphys}
\bibliography{free}

\end{document}